%% file: preprint.tex
\documentclass[twoside]{article}

\usepackage[preprint]{aistats2026}
\usepackage{amsmath}
\usepackage{amsthm}
\usepackage{xparse}
\usepackage{nicefrac}
\usepackage{algorithmic}
\usepackage{lineno}
\usepackage{subcaption}
\usepackage{graphicx}
\usepackage[T1]{fontenc}
\usepackage{tikz}
\usetikzlibrary{arrows.meta, positioning, shapes.geometric}
\usetikzlibrary{calc}
\usepackage{algorithm}
\input{commands}
\usepackage{etoolbox}

\newcommand{\pow}[1]{^{(#1)}}
\newcommand{\lp}{\left(}
\newcommand{\rp}{\right)}
\newcommand{\lc}{\left\{}
\newcommand{\rc}{\right\}}
\newcommand{\lb}{\left[}
\newcommand{\rb}{\right]}
\newcommand{\Cd}{\kappa_n^\dagger}

\usepackage[round]{natbib}

\begin{document}

%
\runningtitle{Multi Change Point Detection for Markov Chains}

%

\twocolumn[

\aistatstitle{Nonparametric Multi Change Point Detection for Markov Chains via \\ Adaptive Clustering}

\aistatsauthor{ Imon Banerjee \And Jiaqi Lei \And  Sanjay Mehrotra }

\aistatsaddress{ Northwestern University \And  Northwestern University \And Northwestern University } ]

\begin{abstract}
  Offline change point detection tries to detect \textit{time points} of distribution change in a given data sequence; and is now routinely used in signal processing, speech processing, climatology etc. Despite this broad applicability across economics, computer science, and planetary sciences, rigorous, nonparametric techniques for change point detection with non-independent and identically distributed (i.i.d.) datasets has remained elusive. This paper establishes such guarantees by proposing a non-parametric clustering algorithm which can accurately obtain the change points from a given Markovian dataset of length $n$. It does so by bridging together two different components of mathematical statistics; Rademacher complexities of Markov chains, and adaptive clustering via penalisation. Our first result uses recent advances in Rademacher complexities of regenerating Markov chains to derive a Dvoretzky Kiefer Wolfowitz (DKW) type inequality for the empirical distribution of the Markov chain. We then use this to show that an adaptive clustering algorithm recovers the correct change points for a Markovian sequence. We establish the tightness of our rates by showing that they essentially coincide with the best known rates for i.i.d. data. We end the paper by discussing the computational considerations of the problem.   
\end{abstract}

\section{INTRODUCTION}\label{sec:introduction}

Change point analysis---due first to the seminal work of Page \citep{Page1954,Page1955}---is a well-established area of study that focuses on identifying points within a data sequence where significant structural changes occur. Yet traditional methods suffer from in two key aspects: (i) the theoretical properties for such methods are most studied when the number of changes are known (and typically a single change), and (ii) the relevant method enjoy the strongest guarantees in the i.i.d. setting, with some oracle based online techniques \citep{bansal1986algorithm,lai1998information,banerjee2012data} used in the dependent-offline context through restarts; whereas most industrial applications (like process engineering, climate dynamics etc.) have dependent (and particularly, Markovian) data streams. This paper addresses both of these shortcomings; we show that without prior knowledge about the number of change points, minimising the \textit{penalised} inter-cluster variance leads to detecting the correct number of change points for Markovian data-sets. 

To set the stage, we introduce some notation. Let $X_1,\dots,X_n$ be samples from a real valued Markov chain on $[0,1]$. Let $\tau_1<\tau_2<\dots <\tau_{K_n}\in\{1,\dots,n\}$ be the true set of change points, with $F_1,\dots,F_{K_n+1}$ being the corresponding stationary distributions. For any given $L$ and any sequence of estimated change points $\tau_1'<\tau_2'<\dots <\tau_{L}'$, the empirical counterparts of $F_i$ is given by $\hat F_{\tau_{i-1}}^{\tau_{i}}(u):=\sum_{j=\tau_{i-1}}^{\tau_i}\indicator[X_j\leq u]/(\tau_i-\tau_{i-1})$. Note that, $K_n$ is allowed to increase with $n$. The risk function we consider is total clustering variance (defined formally in \S \ref{sec:model})
\begin{small}
  \begin{align*}
    \int \sum_{i=1}^L (\tau_i-\tau_{i-1}) \hat F_{\tau_{i-1}}^{\tau_i}(u)\lp1-\hat F_{\tau_{i-1}}^{\tau_i}(u)\rp d\hat F_0^n(u).\numberthis\label{eq:cluster_var}
\end{align*}  
\end{small}
To motivate this risk function, consider its population analogue. As $n\to\infty$, let $\tau_i/n\approx \theta_i$ denote the true (scaled) change points, and let $\hat\theta_i$ be their estimates. Suppose each observation $X_j$ is independently assigned to cluster $i$, that is, drawn from $F_i(u)$ over the interval $(\hat\theta_{i-1},\hat\theta_i]$. Let $\mathcal{X}_i=\#\{X_j: j\in(\hat\theta_{i-1},\hat\theta_i]\}$. Observe that $\Xcal_i$ counts the number of samples in cluster $i$ and is distributed as
\[
\mathcal{X}_i \sim \mathrm{Binomial}(n(\hat\theta_i-\hat\theta_{i-1}),F_i(u)),
\]
Then,
\[
\Var(\mathcal{X}_i)=n(\hat\theta_i-\hat\theta_{i-1})F_i(u)(1-F_i(u)).
\]
Summing across clusters yields
\[
\sum_i \Var(\mathcal{X}_i)=\sum_i n(\hat\theta_i-\hat\theta_{i-1})F_i(u)(1-F_i(u)),
\]
which is precisely the population counterpart of \cref{eq:cluster_var}.


We can understand why the cluster variance is a good metric by looking at the population counterpart. Assume that there is only a single change point $\tau_1$ such that $\tau_1/n\approx\theta$. Now suppose the candidate for the estimated change point is $\tau_1'$ such that $\tau_1'/n\approx \hat \theta$, with the population counterpart of the estimate of $F_1$ being $F_{1,\hat\theta}$ and $F_2$ being $F_{2,\hat\theta}$ as below: 
\begin{small}
\begin{align*}
    F_{1,\hat\theta} & = \dfrac{\min{}(\hat\theta,\theta)F_1(u)+\max{}(\hat\theta-\theta,0)F_2(u) }{\min{}(\theta_0,\theta)+\max{}(\hat\theta-\theta,0)}\qquad \text{ and }\\
    F_{2,\hat\theta} & = \dfrac{\max{}(\theta-\hat\theta,0)F_1(u)+\min{}(1-\hat\theta,1-\theta)F_2(u)}{\max{}(\theta-\hat\theta,0)+\min{}(1-\hat\theta,1-\theta)}.
\end{align*}    
\end{small}
It can then be verified that for each $u$, the clustering variance
\begin{small}
\begin{align*}
    \hat \theta F_{1,\hat \theta} (u)\lp 1-F_{1,\hat \theta} (u)\rp + (1-\hat \theta) F_{2,\hat \theta} (u)\lp 1-F_{2,\hat \theta} (u)\rp
\end{align*}    
\end{small}
decreases as $\hat \theta\rightarrow\theta$ from both sides. 

The pathway towards proving the consistency of the previous mechanism hinges on two main results: (i) A Dvoretzky Kiefer Wolfowitz inequality for regenerating Markov chains (Theorem \ref{prop:talagrand-scaled}), and (ii) the monotonicity of the risk function with respect to spurious detection (Lemma \ref{lemma:monotne}). The proof of the former was by using recent advances in computing the Rademacher complexities of regenerating Markov chains, whereas the second is a consequence of a careful analysis of the risk function (as defined below in \cref{eq:risk}). We discuss ways to further improve our results, which leads to an open question in the Poissonian concentration of Markov chains (further details in \S\ref{sec:nummelin-splitting}).

As mentioned above, change point detection has typically been seen through the lenses of hypothesis testing, based on Kolmogorov-Smirnoff (KS), Cram\'er-von Mises, or Anderson-Darling tests. But in the multi-change point detection, this reveals additional pathologies. A crucial assumption in existing literature attempting to use KS \citep{padilla_optimal_2021} as a baseline for change point detection is access to multiple data streams. In fact, a close inspection of Theorem 3.1 (and in particular equation 3.2) \cite{padilla_optimal_2021} reveals a condition on a tuning parameter, which, in absence of multiple data streams, becomes vacuous, by having a lower bound that is increasing and an upper bound that is an universal constant. 


Finally, we extensively discuss the computational considerations of our method; providing two distinct mixed-integer binary formulations which all recover the correct number of solutions, as opposed to the standing baseline in the field---PELT\citep{killick2012optimal}---which consistently overestimates the correct number of change points.

We now move on to briefly mention our key contributions in this paper. 

\begin{itemize}
    \item \textbf{DKW inequality for regerating Markov chains.} We establish a Hoeffding-style tail inequality for the suprema of the empirical distribution (Dvoretzky Kiefer Wolfowitz inequality) in Theorem \ref{prop:talagrand-scaled}. Corollary \ref{corollary:DKW} establishes that this inequality is essentially sharp; and devolves to the regular DKW inequality for empirical processes for i.i.d. data (up to $\log$ terms). Along the way we also discuss the question of a Bennett-style inequality, which to the best of our knowledge was first identified in 2000 \citep{samson_concentration_2000} but has remained open since, and its implications in this paper (see \S\ref{sec:nummelin-splitting} for more details).

    \item \textbf{Consistent multi-change point detection for Markovian data.} Exploiting Theorem \ref{prop:talagrand-scaled}, we establish the first offline change point detection mechanism for regenerating Markov chains. Our rates are tight, in the sense that when the data is i.i.d., they match the known rates in literature (more details in \S\ref{sec:change-point}).

    \item \textbf{Computational considerations.} 
    We illustrate the effectiveness of the proposed method by comparing it with a common baseline: the Pruned Exact Linear Time (PELT) method \citep{killick2012optimal}. Using simulated nonstationary Markov chain data, we evaluate runtime and detection accuracy of different methods.(see \S\ref{simulated_data}).
\end{itemize}

\textit{The rest of the paper is organised as follows:} \S \ref{sec:background} outlines a comprehensive discussion of relevant research work. We formally introduce the model and relevant notations in \S \ref{sec:model}. \S \ref{sec:main_results} hold our key theorems and the sketches of their proofs, while the full proofs have been deferred to appendix due to lack of space. The computational considerations are housed in \S \ref{sec:computation}, and finally, in \S \ref{sec:conclusions}, we discuss the broader impacts and the future outlooks of our work.


\section{BACKGROUND AND RELATED RESEARCH}\label{sec:background}

Broadly speaking, change point detection methods can be categorised as \emph{online} (detecting changes in real‑time) or \emph{offline} (segmenting retrospectively once all data are observed), the latter of which is the focus of this paper. Online methods prioritize early detection, while offline methods typically aim to identify multiple changes simultaneously.

The earliest contribution on change point detection in the Gaussian setting are due to \cite{Page1954,Page1955}. Since then, change point methodology has been applied in speech processing \citep{desobry2005online,Harchaoui2009}, finance \citep{Bai1998a,Frick2014}, bioinformatics \citep{Hocking2013,Maidstone2017}, climatology \citep{maidstone2016efficient,Verbesselt2010}, and network traffic analysis \citep{Levy-Leduc2009,Lung-Yut-Fong2012}, among others. We point the reader to standard references like the monographs \cite{Basseville1993,Darkhovsky1993,Csorgo1997,chen2011parametric} and the surveys \cite{truong_selective_2020,Lavielle2007,Jandhyala2013,Haynes2017} for more details on the theory and methods of change point analysis.

Reliable detection of multiple changes is a comparatively recent development. A key contribution is \cite{zou_nonparametric_2014} (see also \cite{fryzlewicz2014}), who proposed a nonparametric maximum likelihood method with BIC-based model selection, establishing consistency and optimal rates. More recently, \cite{padilla_optimal_2021,madrid_padilla_optimal_2022} analyzed univariate and multivariate data using the Kolmogorov-Smirnov (KS) distance, showing that wild binary segmentation achieves nearly minimax rate-optimality in multi-stream settings.

However, possibly due to the absence of suitable mathematical tools \citep{bertail_rademacher_2019}, the problem of offline nonparametric multiple change point detection for Markov chains remains open, with limited progress under parametric settings such as time-series models \citep{Fryzlewicz2017}. We get the following open question.
\paragraph{Open Question.} Is it possible to design a \emph{nonparametric change point detection method for Markov chains} on general state spaces with mild assumptions, and achieving optimal convergence?
\begin{remark}
The choice of a nonparametric approach in this paper is motivated primarily due to the fact that apart from a few developments in time-series and queueing contexts, parametric models for Markov chains remain largely undeveloped.
\end{remark}
\section{PROBLEM STATEMENT}\label{sec:model}
\newcommand{\argmin}{\operatorname{argmin}}
\newcommand{\argmax}{\operatorname{argmax}}
In order to formalize our results, we introduce some notations. A $\sigma$-algebra $\Ecal:=\sigma(E)$ of a set $E$ is said to be countably generated if there exists a countable set $\Ccal$ such that $\Ecal\subseteq\sigma(\Ccal)$. Then, we call $(E,\Ecal)$ to be a countably generated state space. Our choice of $(E,\Ecal)=([0,1],\Bcal_{[0,1]})$, with $\Bcal_{[0,1]}$ denoting the Borel $\sigma$-algebra. In particular, $([0,1],\Bcal_{[0,1]})$ is countably generated. We define the Orlicz norm of a random variable $X$ as $\|X\|_{\psi_1}=\argmin  \{\lambda>0 : \expec\lb e^{(X/\lambda)}\rb\leq 1\}$.

\textbf{Data generating process.} Let us assume that we have data that is a sample from Markov chains, defined on a continuous state space $[0,1]$, that change at $K_n$ (unknown) time points. The transition densities and invariant distributions of these Markov chains are denoted as $\density_1,\density_2,\dots,\density_{{K_n+1}}$;  $F_1,F_2,\dots,F_{K_n+1}$, respectively.  We have $X^1_{1},\dots,X^1_{n_1}$ samples from the first transition density, $X^2_{1},\dots,X^2_{n_2}$ samples from the second transition density, and so on. In the offline change point detection problem, we do not have access to $n_1$, $n_2$, etc. and observe our data as a single sample $X_1,\dots,X_n$ with $n=n_1+n_2+\dots$, and
\begin{equation}
    \begin{split}
        X_1=X_1^1,X_2=X_2^1,\dots, X_{n_1}=X_{n_1}^1, \\
        X_{n_1+1}=X_1^2,\dots,X_{n_1+n_2}=X_{n_2}^2,\dots.
    \end{split}\label{eq:triangular}   
\end{equation}

With $\tau_1=n_1,\tau_2=n_1+n_2,\dots$ as the reparametrizations of the true (but unknown) change points, our problem then is to estimate the $\tau_i$'s from a single data-stream $X_1,\dots,X_n$. We recall from the introduction that $\tau_1<\tau_2<\dots <\tau_{K_n}$ is the ordered set of true change points, and remind the readers that $K_n$ is allowed to increase with $n$. The empirical counterparts of $F_i$ is given by $\hat F_{\tau_{i-1}}^{\tau_{i}}(u):=\sum_{j=\tau_{i-1}}^{\tau_i}\indicator[X_j\leq u]/(\tau_i-\tau_{i-1})$ and let $\hat F_n(u):=\sum_{i=1}^n\indicator[X_i\leq u]/n$ the empirical distribution of the whole sample.

Let $L$ be an integer estimating the number of change points with the candidate change points being $\tau_1',\dots,\tau_L'$. For any $u\in[0,1]$ and $\tau_i'<\tau_j'$, we define $\hat F_{\tau_i'}^{\tau_j'}:=\sum_{p=\tau_i'}^{\tau_j'}\indicator[X_p\leq u]/{(\tau_j'-\tau_i')}$ as the empirical distribution between $\tau_i'$ and $\tau_j'$. By $X_{(1)},\dots, X_{(n)}$, we denote the order statistics of $X_1,\dots,X_n$.

Our choice of loss metric is the nonparametric clustering variance defined as: 
\begin{small}
  \begin{align}
&R_n(\tau_1',\dots,\tau_L')\nonumber\\ 
    &= \sum_{i=0}^{L-1}(\tau_{i+1}'-\tau_{i}') \int_{X_{(1)}}^{X_{(n)}} \hat{F}_{\tau'_i}^{\tau'_{i+1}}(u)(1-\hat{F}_{\tau'_i}^{\tau'_{i+1}}(u))d\hat F_n(u)\label{eq:risk}
\end{align}  
\end{small}
When $\tau_1',\dots,\tau_L'$ is not ordered, then 
\begin{align*}
    R_n(\tau_1',\dots,\tau_L') := R_n\lp\tau_{(1)}',\dots,\tau_{(L)}'\rp\numberthis\label{eq:risk_unordered}
\end{align*}
where $\tau_{(1)}'<\tau_{(2)}'<\dots<\tau_{(L)}'$ is the order statistics for $\tau_1',\dots,\tau_L'$.

Our objective is the risk penalised by the Bayesian information criterion
\begin{align*}
    \mathrm{BIC_L} := \min_{\tau_1'<\dots<\tau_L'} R_n(\tau_1',\dots,\tau_L') + L\zeta_n,\tag{BIC} \label{eq:bic}
\end{align*}
where $\zeta_n$ is a suitably increasing sequence given explicitly in Theorem \ref{thm:marginal}. Our choice of estimators will be 
\begin{align*}
    \hat K_N & := \argmin_L \mathrm{BIC_L} \qquad \text{ and } \\\Gcal_n(L) & :=(\hat \tau_1,\dots\hat \tau_{L})=\argmin_{\tau_1'<\dots<\tau_L'} R_n(\tau_1',\dots,\tau_L'),
\end{align*}
with $\Gcal_n(\hat K_n)$ being the final estimated change points.
\begin{remark}
    Observe that solving \cref{eq:bic} by enumeration is NP-hard. In \S \ref{sec:computation}, we provide a mixed integer program to minimise this objective.
\end{remark}

Theorem \ref{thm:marginal}, proves that optimising the previous objective function achieves the optimal rate of detecting change points. 

Before proceeding we briefly justify our choice of integrating measure $d\hat F_n(u)$ in \cref{eq:risk}. Prior work \citep{zhang_powerful_2002,zhang_powerful_2006} highlights the empirical benefit of using weighted measures. In nonparametric change-point analysis, $d\hat F_n(u)$ is effective when adjacent distributions differ near their medians. However, as noted by \cite{zou_nonparametric_2014} in the i.i.d.\ setting, it can underperform when differences lie in the tails, since little information is captured in the integral. In such cases, using 
\[
\frac{d\hat F_n(u)}{\hat F_n(u)(1-\hat F_n(u))}
\]
offers greater detection power. This improvement, however, relies on a Poissonian concentration inequality for i.i.d.\ processes \citep{wellner_limit_1978}, with no clear analogue in the Markovian case. Deriving a Markovian counterpart to Lemma~1 of \cite{wellner_limit_1978} appears infeasible without new tools, so we instead adopt the unweighted integrating measure which is amenable to the DKW inequality. We emphasize that such weighting only improves empirical power: Theorem~\ref{thm:marginal} already guarantees the asymptotically optimal rate. We now turn to the formal statements of our main results.





\section{THEORETICAL RESULTS}\label{sec:main_results}

As stated in the introduction, our first objective will be to provide a DKW inequality for the empirical suprema of a regenerating Markov chains. To that end, we provide a brief introduction to regenerating Markov chains. 
 
 \subsection{DKW Inequality for Regenerating Markov chains}\label{sec:nummelin-splitting}


We give the following definition of atomic $\psi$-irreducible Markov chain (we refer the readers to  \cite[page 89]{meyn_markov_2012} for the formal definition of $\psi$-irreducibility).

\begin{definition}[Regenerative/Atomic chain]\label{def:atomic}
A \(\psi\)-irreducible, aperiodic Markov chain \(X\) with a transition probability distribution $P(\cdot,\cdot)$ is \emph{regenerative} (or \emph{atomic}) if there exists a measurable set \(A\) (the \emph{atom}) with \(\Psi(A)>0\) (for some measure $\Psi$) such that
\[
P(x,\cdot)=P(y,\cdot),\qquad \forall\,x,y\in A.
\]
\begin{remark}
    The set $A$ is called the $\Psi$-atom. In chains with finitely many states any single state may serve as an atom.
\end{remark}
\end{definition}
 Intuitively, $\psi$-irredicibility extends the classical notion of irreducibility for finite state Markov chains to infinite state spaces, whereas $\Psi$-atoms are sets from which the transitions behave homogeneously. That is, the probability of transition to any set is equal for any two starting points inside an atom. For finite state Markov chains, the atoms are the individual states (i.e. singletons). Informally, a regenerating Markov chain is a $\psi$-irreducible Markov chain which has at least one $\Psi$-atom that is repeatedly visited (with the inter-arrival times termed as the regeneration time). Conditions on the moment generating function of the regeneration time characterises the ergodic properties of the Markov chain as can be seen in Assumption \ref{assume:regenerate} below.

We now formalise the previous intuition. Extend the sample space by introducing a sequence $(Y_m)_{m\in\mathbb{N}}$ of independent Bernoulli random variables with success probability $\delta$. Our construction relies on a mixture representation of the transition kernel on a set $S$: 
\[
P(x,S)=\delta\,\Psi(S)+(1-\delta)\frac{P(x,S)-\delta\,\Psi(S)}{1-\delta},
\]
where the first term is independent of the starting point. The validity of this construction is ensured by the existence of the atom and we refer to \cite[Chapter 4]{meyn_markov_2012} for more details. This independence is key since it guarantees regeneration when that component is selected. In other words, each time the chain visits $S$, we randomly reassign the transition probability $P$ as follows:
\begin{itemize}
    \item If $X_m\in S$ and $Y_m=1$ (which occurs with probability $\delta\in(0,1)$), then the next state $X_{m+1}$ is generated according to the measure $\Psi$.
    \item If $X_m\in S$ and $Y_m=0$ (with probability $1-\delta$), then $X_{m+1}$ is drawn from the probability measure 
    \[
    (1-\delta)^{-1}\Bigl(P(X_m,\cdot)-\delta\,\Psi(\cdot)\Bigr).
    \]
\end{itemize}

The resulting bivariate process
\[
Z_m=(X_m,Y_m),
\]
known as the \textit{split chain}, takes values in $E\times\{0,1\}$ and is itself atomic, with the atom defined as
\[
A=S\times\{1\}.
\]
We then define the \textit{regeneration times} recursively by setting
\[
\rho_{A}(1)=\inf\{m\geq 1: Z_{m}\in A\},
\]
and for $j\geq 2$,
\[
\rho_{A}(j)=\inf\{m>\rho_{A}(j-1): Z_{m}\in A\}.
\]

It is well known that the split chain $Z$ inherits aperiodicity and $\psi$-irreducibility from the original chain $X$. Furthermore, by the recurrence property, the regeneration times have finite expectation; that is, one has \citep[Lemma A1]{AzaisDelyonPortier2016}
\[
\sup_{z\in A}\mathbb{E}_{z}[\rho_{A}(j)]<\infty\quad\text{and}\quad \mathbb{E}_{\nu}[\rho_{A}(j)]<\infty
\]
for any initial measure $\nu$ on $A$ and any integer $j$.

Regeneration theory \citep{meyn_markov_2012} shows that, given the sequence $(\rho_{A}(j))_{j\geq1}$, the sample path can be divided into blocks (or cycles) defined by
\[
B_{j}=(X_{\rho_{A}(j)+1},\dots,X_{\rho_{A}(j+1)}),\quad j\geq2,
\]
corresponding to successive visits to the regeneration set $A$. The strong Markov property then ensures that both the regeneration times and the blocks $\{B_{j}\}_{j\geq1}$ form independent and identically distributed (i.i.d.) sequences \citep[Chapter 13]{meyn_markov_2012}. Moreover, with $\mathbb{E}_{A}$ to be the expectation when the Markov chain is initialised from the atom $A$, the lengths of these blocks are i.i.d. with mean $\mathbb{E}_{A}[\rho_{A}(2)]$. Furthermore, $B_1$ is independent of $B_2,B_3,\dots$.


\begin{assumption}\label{assume:regenerate}
     There exists a constant $\lambda>0$ for which
    \begin{equation}\label{def_Ea}
        \mathbb{E}_{A}[\exp(\lambda \rho_{A}(2))]<\infty.
    \end{equation}
    Moreover, for this choice of $\lambda$, we                                                       assume that 
    \begin{equation}\label{def_clambda}
   \kappa_{\lambda}:=\frac{2\,\mathbb{E}_{A}[\exp(\lambda \rho_{A}(2))]}{\lambda}> \frac{1}{2}.     
    \end{equation}
\end{assumption}

Assumption 1 ensures that the regeneration time has an exponential moment, a property equivalent to geometric ergodicity, the uniform Doeblin condition, or the Foster-Lyapunov drift criterion (see Theorem 16.0.2 in \cite{meyn_markov_2012}). Under this assumption, classical convergence results such as the LIL and CLT remain valid (Chapter 17). Next, we define 
 \begin{equation}\label{def_rho_o}
  \rho_o:=\min\{ \|\rho_A(1)\|_{\psi_1},\|\rho_A(2)\|_{\psi_1} \}.   
 \end{equation}
 Observe that under  Assumption \ref{assume:regenerate}, the sets
{\small \begin{align*}
    \{\lambda>0: \expec\lb e^{(\rho_A(1)/\lambda)}\rb\leq 1\} 
    \text{ and }  \{\lambda>0 : \expec\lb e^{(\rho_A(2)/\lambda)}\rb\leq 1\}
\end{align*}}
are non-empty. It follows that $\rho_o<\infty$. 

We now proceed to formally state the main result of this section. Let $\Scal\subset \Rbb$ and define $\Fcal_\Scal:=\{ \indicator[\cdot<t], t\in \Scal\}$ be the set of all half-interval functions on $\Scal$. The following theorem is proved in \S \ref{sec:prf-tal-scaled}.

\begin{theorem}\label{prop:talagrand-scaled}
    Let $Y_1,\dots,Y_n$ be a sequence of random variables from a Markov chain satisfying Assumption \ref{assume:regenerate} with stationary distribution $\pi$. Let $Y$ be a random variable with distribution $\pi$. Define 
    \begin{align*}
        Z:= \sup_{f\in\Fcal_{[0,1]}}\lv\frac{1}{n}\sum_{i=1}^n f(Y_i) - \expec_{\pi}\lb f(Y)\rb\rv.
    \end{align*}
    Then, for all $t>0$
    \begin{align*}
    \prob(Z>t) \leq \kappa(\rho,\lambda)\exp\lp -\frac{\kappa(\rho,\lambda)nt^2}{\log n} \rp\numberthis\label{eq:talagrandscaled-conc}
    \end{align*}
    where, constant $\kappa(\rho,\lambda)$ depend only on $\expec_A[\rho_A^2(2)]$, $\expec_v[\rho_A(2)]$, $\rho_o$, and $\lambda$ (as defined in Eq. \eqref{def_Ea}, \eqref{def_clambda}, and \eqref{def_rho_o}).
\end{theorem}
\begin{remark}
    The extra $\log n$ term appears commonly when transitioning from i.i.d random variables to Markov chains \citep{samson_concentration_2000}. Observing that the regeneration time $\rho$ of any i.i.d. process is (deterministically) $1$, the following corollary shows that our previous result is tight up to $\log$ terms.
\end{remark}
\begin{corollary}\label{corollary:DKW}
       Let $Y_1,\dots,Y_n$ be a sequence of i.i.d. random variables from distribution $\pi$. Let $Y$ be a random variable with distribution $\pi$. Then, with $Z$ defined as in Theorem \ref{prop:talagrand-scaled}, and an universal constant $\constant$; 
    \begin{align*}
    \prob(Z>t) \leq \constant\exp\lp -\frac{\constant nt^2}{\log n} \rp\numberthis\label{eq:talagrandscaled-conc}
    \end{align*}
\end{corollary}
\paragraph{On Poissonian Tail Concentration:} Observe that Theorem \ref{prop:talagrand-scaled} provides a H\"oeffding (sub-Gaussian) concentration inequality (ignoring constants)
\begin{align*}
    \prob(Z_n>t)\lesssim  e^{-nt^2/\log n}
\end{align*}
for the suprema of the empirical distribution, which is sharp when $t$ is large. However, for small values of $t$, a sharper bound can be obtained from so called Bennett (Poissonian) concentration inequalities
\begin{align*}
    \prob(Z_n>t)\lesssim  e^{-nt\log (1+t)}.
\end{align*}
Owing to the fact that such inequalities for the empirical suprema of i.i.d. processes were known due to Michael Talagrand, \citep{Talagrand1996}, our initial objective in this paper was to derive a Bennett inequality for regenerating Markov chains and then use it to derive the consistency of the \textit{weighted} risk function described in \S\ref{sec:model}. However, upon further investigation, multiple papers \citep{samson_concentration_2000,adamczak_tail_2008} point out that such a result requires novel approaches towards empirical process theory, which makes it beyond the scope of the current work. The particular pathology is the unavailability of the Hoffmann-Jørgensen inequality (more specifically, a counterpart to Theorem 6.21 in \citeauthor{ledoux_probability_1991}) for the Bennett-Orlicz norm \citep{wellner_bennett-orlicz_2017}. However, we point out that availability of any such result will be immediately applicable; and proving the consistency of the weighted risk function would be a straightforward extension of the proof in \S\ref{sec:prf-marginal}.
\begin{remark}
    The previous challenge is not purely theoretical, since it clearly leads to a pathway of improving the risk function for change point detection of Markov chains.
\end{remark}

\subsection{Change Point Detection}\label{sec:change-point}
Now we proceed to state the main result of this section. To that end, we make some assumptions, beginning with the regeneration times of the underlying Markov chains. Recall from \cref{eq:triangular} that $X_j^i$ was used to denote the $j$-th sample of the $i$-th Markov chain, with $K_n+1$ many Markov chains being available in total. 
\begin{assumption}\label{assume:regenerate-change}
    All of the Markov chains $X_j^i$ satisfy Assumption \ref{assume:regenerate} with constants $\kappa_i$ for $i\in \{1,\dots,K_n+1\}$
\end{assumption}
\begin{remark}
    Observe it is not sufficient to simply impose Assumption \ref{assume:regenerate} the Markov chain  $X_1,\dots,X_n$, since one of $X_j^i$'s may be transient while the whole chain stays regenerating.
\end{remark}
We next make the following standard assumption (see \citet[A1-A3]{zou_nonparametric_2014}) of sufficient gap between any two consecutive changes:

\begin{assumption}\label{assume:consistency}
    Let $\beta_n : = \min_{1\leq k\leq K_n+1}\lp \tau_k-\tau_{k-1} \rp$. We assume that
    \begin{align*}
       \beta_n \xrightarrow{n\rightarrow\infty}\infty \text{ \quad and \quad } \hat F_n(u)\xrightarrow{n\rightarrow\infty} F(u)
    \end{align*}
    almost everywhere uniformly in $u$, for some distribution $F$ in the convex hull generated by $\{F_1,\dots,F_{K_n+1}\}$.
\end{assumption}
Finally, for all $r\in \{1,\dots,K_n+1\}$ define $\eta_r(u):= (F_{r-1}(u)-F_{r}(u))^2$. Our next assumption is on the minimum amount of change between any two $F_i$'s. It can be thought of as the minimum signal strength indicating change.
\begin{assumption}\label{assume:identifiability} 
$F_{k}$ are continuous and distinct for all $k$ and define
    \begin{align*}
        \eta_{\min{}} := \min_{1\leq r\leq K_n+1} \int_0^1 \eta_r(u)dF(u).
    \end{align*}
    Note that, $\eta_{\min{}}>0$.
\end{assumption}
\begin{remark}
    We assume that we know the true value of $\eta_{\min{}}$. In practice, we can replace it by a large enough constant. By a careful inspection of the proof of Theorem \ref{thm:marginal} (and in particular \cref{eq:onetamin}) one can see that any constant larger than $\eta_{\min}$ is actually sufficient for detecting the correct number of change points.
\end{remark}

To formalise our main result we introduce some notation. Let $L>0$ be an integer and define $\tau_1',\dots,\tau_{K_n}'$ be a sequence of time points. We define the set $\Ccal_{K_n}(\delta_n)$ as the set of all $(\tau_1',\dots,\tau_{K_n}')$ which are at most $\delta_n$ away from the true change point $\tau_1,\dots,\tau_{K_n}$. Formally,
\begin{align*}
    \Ccal_{K_n}(\delta_n) & := \{(\tau_1',\dots,\tau_{K_n}'):1<\tau_1'<\dots<\tau_{K_n}'\leq n, \\
    & \qquad |\tau_s'-\tau_s|\leq \delta_n \ \forall \  1\leq s\leq K_n \}
\end{align*} 
for some positive sequence $\delta_n$. Next, for a fixed $L$, define $\Gcal_n(L)$ as the set of time points which minimizes $\operatorname{BIC_L}$. Formally,
\begin{align*}
    \Gcal_n(L)
    = & (\hat \tau_1,\dots , \hat\tau_L):=\argmin_{\tau_1'<\dots<\tau_L'} \operatorname{BIC_L}\\
    = & -\argmax_{\tau_1'<\dots<\tau_L'} R_n(\tau_1',\dots,\tau_L')+L\zeta_n.
\end{align*}
Then, we have the following theorem, proved in \S\ref{sec:prf-marginal}.

\begin{theorem}\label{thm:marginal}
    Under Assumptions \ref{assume:regenerate}, \ref{assume:consistency}, and \ref{assume:identifiability}, and  if ${K_n^3{(\log K_n)^2(\log \delta_n)^2}/{\delta_n}}=\Ocal(1)$ and $\delta_n/\beta_n\rightarrow0$, then with any sequence $\zeta_N\rightarrow\infty$, we have
   \begin{align*}
        \probl\lp \Gcal_n(K_n)\in \Ccal_{K_n}(\delta_n) \rp\xrightarrow{n\rightarrow\infty\ }1.
    \end{align*}
    Furthermore, let $K_n\leq \bar K$ for some known $\bar K$. Then, with $\zeta_n = {\kappa_n\bar K^3{(\log \bar K)^2(\log \beta_n)^2}/{\beta_n}}$,
    \begin{align*}
        \hat K_n :=\arg\min_L \Gcal_n(L)
    \end{align*}
    satisfies $\prob\lp \hat K_n = K_n \rp\rightarrow 1$, where $\kappa_n$ is a constant that depends only on $\kappa_i, i\in\{1,\dots,K_n+1\}$ .
\end{theorem}
Theorem \ref{thm:marginal} provides the asymptotic consistency of $\hat K_n$, which is the estimated number of change points. It implies that the localization accuracy $\delta_n$ may be taken as any sequence satisfying
\[
K_n^3 (\log K_n)^2 (\log \delta_n)^2 / \delta_n = O(1), \qquad \delta_n / \beta_n \to 0.
\]
In particular, if $K_n = O(1)$, then $\delta_n = O(1)$ is admissible, yielding constant-error localization—the best achievable rate in nonparametric change-point detection. When $K_n$ grows, the attainable $\delta_n$ remains of order $\beta_n$ up to polylogarithmic factors, matching known optimal rates in the i.i.d.\ literature \citep{zou_nonparametric_2014}.

Now consider the independent case. Here the regeneration time is $1$, with the atom equal to the entire state space. Hence any i.i.d.\ sequence satisfies Assumption~\ref{assume:regenerate-change} with deterministic regeneration time $1$ ($\rho_A(j)=1$ for all $j$).
 In particular, $\kappa_n$ in the Theorem \ref{thm:marginal} becomes an universal constant, and we recover the following corollary.\begin{corollary}\label{corollary:iid}
    Let $X_i$ be independent. Then, under Assumptions \ref{assume:regenerate}, \ref{assume:consistency}, and \ref{assume:identifiability}  if ${K_n^3{(\log K_n)^2(\log \delta_n)^2}/{\delta_n}}=\Ocal(1)$ and $\delta_n/\beta_n\rightarrow0$, then with any sequence $\zeta_N\rightarrow\infty$, we have
    \begin{align*}
        \probl\lp \Gcal_n(K_n)\in \Ccal_{K_n}(\delta_n) \rp\xrightarrow{n\rightarrow\infty\ }1.
    \end{align*}
    Furthermore, let $K_n\leq \bar K$ for some known $\bar K$. Then, with $\zeta_n ={\bar K^3{(\log \bar K)^2(\log \beta_n)^2}/{\beta_n}}$,
    \begin{align*}
        \hat K_n :=\argmin_L \Gcal_n(L)
    \end{align*}
    satisfies $\prob\lp \hat K_n = K_n \rp\rightarrow 1$.
\end{corollary}
Before presenting a sketch of the proof of the theorem, we discuss some aspects of the previous result.
\paragraph{Tightness:} A direct comparison of Corollary \ref{corollary:iid} with \citet[Theorem 1]{zou_nonparametric_2014} reveals that the rate of $\beta_n$ do not suffer when we move from i.i.d. to Markovian data, but rather adjusts for the Markovianity in the obvious way, through a constant which depends upon the regeneration time, and is an universal constant for i.i.d. data. 

\paragraph{Choice of penalty:} We further remark that $\beta_n$ is typically unavailable to the practitioner, and setting the penalty $\zeta_n=(\log n)^{2+c}$ for a small constant $c\geq 0$ seems to work well in practice \citep{zou_nonparametric_2014}. We use $c=0$ in \S\ref{sec:computation}. We now provide a sketch of proof of Theorem \ref{thm:marginal}.

\textbf{Sketch of Proof of Theorem \ref{thm:marginal}}
To guide the readers through the broad steps of the proof, we provide a sketch of the proof. For the full details, see \S\ref{sec:prf-marginal}.
\begin{enumerate}
    \item Use Theorem \ref{prop:talagrand-scaled} to develop a deviation bound for the clustering variance risk (Lemma \ref{lemma:KLD-bound}).
    \item Use Lemma \ref{lemma:KLD-bound} to bound the rate of growth of the risk $R_n$ of detecting superfluous change points (Lemmas \ref{lemma:monotne} and \ref{lemma:monotone3}).
    \item Use Lemma \ref{lemma:KLD-bound} to bound the rate of growth of the clustering variance of detecting superfluous change points (Lemma \ref{lemma:monotone2})
    \item Use Lemmas \ref{lemma:monotne}, \ref{lemma:monotone3}, and \ref{lemma:monotone2} to bound the probability of detecting less than the correct number of change points (Lemma \ref{lemma:BIC-minimalrecovery}).
    \item Using contradiction, bound the probability of detecting more than the correct number of change points (Lemma \ref{lemma:BIC-exactrecovery}). Theorem \ref{thm:marginal} follows by combining Lemmas \ref{lemma:BIC-minimalrecovery} and \ref{lemma:BIC-exactrecovery}.
\end{enumerate}

\section{COMPUTATION}\label{sec:computation}
We recast the risk in \eqref{eq:risk} as a nonlinear binary optimization problem. Let $z_{i,l}\in\{0,1\}$ indicate whether time point $i$ is assigned to segment $l\in\{1,\dots,L+1\}$. The segmentation constraints enforce single assignment, contiguity, and a minimum segment length. The resulting formulation both encodes \eqref{eq:risk} exactly (for fixed $L$) and readily accommodates additional structural constraints that are awkward for classical dynamic programming. We therefore write the following proposition, whose proof can be found in \S\ref{sec:prf-optimisation}.

\begin{proposition}\label{prop:optimization}
Let $a_{i,u}:=\indicator[X_i\leq X_{(u)}]$  and
\begin{align*}
    \mathcal{Z}_{i,l}& :=\Big\{z_{i,l}\in\{0,1\}\ \Big|\ \sum_{l=1}^{L+1}z_{i,l}=1,\ \sum_{i=1}^{n}z_{i,l}\ge 3,\\ &\qquad \ z_{i,l}\le \sum_{l'\ge l} z_{i+1,l'}\Big\}.
\end{align*}

Consider
\begin{equation}\label{nonlinear_model}
\begin{aligned}
\min\quad & \sum_{l=1}^{L+1}\sum_{u=1}^{n}\bigg(\sum_{i=1}^{n} a_{i,u}z_{i,l}\bigg)\left(1-\frac{\sum_{i=1}^{n} a_{i,u}z_{i,l}}{\sum_{i=1}^{n} z_{i,l}}\right)\\
\text{s.t.}\quad & z_{i,l}\in \mathcal{Z}_{i,l}\qquad \forall i,l.
\end{aligned}
\end{equation}
Then $z^\star$ is a solution of \eqref{nonlinear_model} if and only if the induced change points $\tau'_l-\tau'_{l-1}=\sum_{i=1}^{n} z^\star_{i,l}$ solve \eqref{eq:risk}; moreover, the objective values of \eqref{nonlinear_model} and \eqref{eq:risk} coincide.
\end{proposition}
To reduce runtime, we use the following bilinear reformulation, formalized in Proposition \ref{prop:bilinear}, which preserves optimality (see \S \ref{sec:prf-optimisation_bilinear} for the proof) while avoiding division inside the objective, and can be solved with off-the-shelf solvers like Gurobi. A detailed empirical comparison of runtimes is provided in \S\ref{simulated_data}.  
\begin{proposition}
\label{prop:bilinear}

Consider the bilinear reformulation:
{\small
\begin{equation}
\label{linear_model}
\begin{aligned}
\min\quad & \sum_{l=1}^{L+1}\sum_{u=1}^n s_{u,l} \\[2pt]
\text{s.t.}\quad
& \sum_{i=1}^n k_l  z_{i,l}  = 1  \quad \forall l,\\
& f_{u,l} = \sum_{i=1}^n a_{i,u} k_l  z_{i,l}, \text{ and} \ 
 d_{u,l} + f_{u,l} = 1 \  \forall u, l,\\
& s_{u,l} = \sum_{i=1}^n a_{i,u} d_{u,l} z_{i,l} \ \forall u, \forall l; \quad z_{i,l}\in \mathcal{Z}_{i,l} \quad  \forall i,\forall l,\\
& k_l,\ f_{u,l},\ d_{u,l},\ s_{u,l}\geq 0\quad \forall u, \forall l,  
\end{aligned}
\end{equation}
}
 then $z^*$ is a solution of \eqref{nonlinear_model} if and only if $(z^*,k^*,f^*,d^*,s^*)$ is a solution of \eqref{linear_model} and the objectives are the same.  
\end{proposition}
\subsection{Simulation}\label{simulated_data}
We compare \eqref{nonlinear_model}, its bilinear reformulation \eqref{linear_model}, and PELT on a nonstationary Markov chain. We generate $n=250$ time points partitioned into $4$ segments of lengths $0.1n,0.2n,0.3n,0.4n$ (with the corresponding time points $\tau_i$ being 25,75,150 respectively). On segment $l$, arrivals and departures follow the arrival rate $\lambda_l$ and departure rate $\mu_l$ are drawn from uniform distributions:
\[
\lambda_l \sim \text{Uniform}(0,\alpha_\lambda), 
\quad 
\mu_l \sim \text{Uniform}(0,\alpha_\mu),
\]
where $\alpha_\lambda,\alpha_\mu > 0$ are scaling factors controlling variability.  
At each time point $i$, let $l(i)$ denote the segment containing point $i$. The system evolves according to
\[
A_i \sim \text{Poisson}(\lambda_{l(i)}), 
\quad 
D_i \sim \text{Binomial}(N_{i-1}, \mu_{l(i)}),
\]
The population is updated recursively by adding arrivals and subtracting departures.
Be specific,
$
N_i = \max\{0,\, N_{i-1} + A_i - D_i\}.
$
All PELT results are obtained using the changepoint.np package, which implements the nonparametric PELT method of \cite{zou_nonparametric_2014}. Experiments were run using the default MBIC penalty, which provided the strongest empirical performance. Note that PELT guarantees pruning efficiency (\cite{killick2012optimal}, Thm. 3.1) but not correctness; in contrast, our MIP formulations exactly minimize the clustering-variance risk.

Figure~\ref{fig:nonMarkov} shows that our method recovers the true change points exactly, $\hat\tau=[25,75,150]$, while PELT over-segments ($\hat\tau=[25,37,46,72,151,161,176,204]$). Across $n$, \eqref{linear_model} matches the accuracy of \eqref{nonlinear_model} at substantially lower cost, whereas PELT is fastest but consistently overestimates $L$ (Table~\ref{tab:placeholder}).

\begin{figure}
    \centering
    \includegraphics[width=\linewidth]{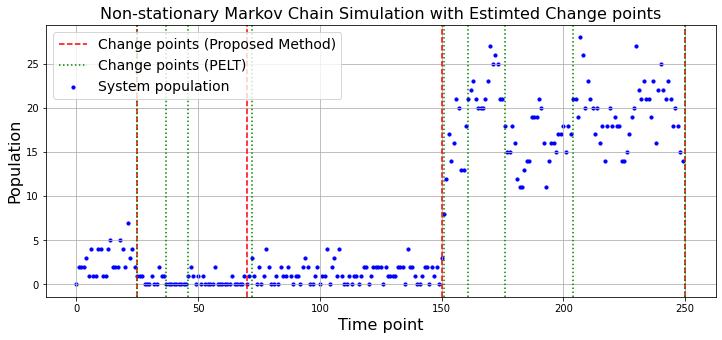}
    \caption{\small Illustration of the simulated nonstationary Markov chain data. Blue dots represent the population size at each time point. Red lines (our proposed method) mark $\hat{\tau}=[25, 75, 150]$, which match the true change points. Green lines (the PELT method) mark $\hat{\tau}= [25, 37, 46, 72, 151 ,161, 176, 204]$, which overestimate the total number of change points.}
    \label{fig:nonMarkov}
\end{figure}

We now compare runtime and accuracy. Across all $n$, both \eqref{nonlinear_model} and \eqref{linear_model} exactly recover the true change points, with the bilinear reformulation achieving a substantial speedup over \eqref{nonlinear_model}. PELT is the fastest but consistently over-segments. Thus, \eqref{linear_model} offers the best trade-off, retaining accuracy while reducing runtime, whereas PELT sacrifices reliability for speed (Table~\ref{tab:placeholder}).

{
\begin{table}[!h]
    \centering
    \caption{\small Runtime and change points of the optimization frameworks and the PELT method with different numbers of time points $n$. \eqref{nonlinear_model}-Time (seconds) is the runtime for \eqref{nonlinear_model}, which is slowest, while \eqref{linear_model}-Time (seconds) is the runtime for \eqref{linear_model}, which is much faster; both recover the true change points. PELT-Time (seconds) is the runtime and PELT-L is the number of change points estimated by the PELT method. The PELT method is the fastest but consistently estimates more than the correct number of change points.}
    \begin{small}
          \begin{tabular}{|c|cc|cc|}
    \hline
     n  &  
   \eqref{nonlinear_model}-Time & \eqref{linear_model}-Time   &  PELT-Time & PELT-L \\
   \hline
   50   & 5.49 & 0.69  & 0.03 &  7 \\ 
   100 & 9.98 &1.42  &0.07  & 8\\
   250&30.42  & 9.43  & 0.35  & 9\\
   400 & 180.38  & 28.71 & 1.14  &  22 \\
   500& 302.76  & 92.59  & 1.53  &  26\\
   \hline
    \end{tabular}
    \end{small}
    \label{tab:placeholder}
\end{table}}
\section{CONCLUSION}\label{sec:conclusions}
Clustering for Markov chains has some recent developments \citep{lee2025near}, and we present a nonparametric framework for multiple change-point detection in Markov chains that couples a DKW-type inequality for regenerating chains with an adaptive clustering criterion. The theory yields consistent estimation of both the number and locations of changes with rates that match the i.i.d.\ benchmark up to logarithmic factors, indicating that dependence need not fundamentally hinder detectability. 

On the computational side, we provide exact optimization formulations that achieve high accuracy; the bilinear model \eqref{linear_model} offers a practical speed–accuracy trade-off and outperforms a PELT baseline that tends to over-segment in our experiments. 

Two directions appear especially promising. First, Bennett–type (Poissonian) concentration for Markov chains would directly enable weighted risks with greater power in tail-difference regimes (see \S\ref{sec:nummelin-splitting}). Second, enriching the optimization with robustness and domain constraints (e.g., minimum dwell times, forbidden transitions, or outlier resistance) would broaden applicability without sacrificing exactness.

\paragraph{Limitations and future outlook.} Online methods such as CUSUM require knowledge of the pre- and post-change distributions and use their KL-divergence to set thresholds. Offline method addresses the complementary regime where these distributions are unknown, trading immediate detection for identifiability. As we note earlier, online methods with unknown pre- and post-change distributions being an evolving field.

We also remark that the multivariate extension remains open, since meaningful generalizations likely require kernel-based empirical processes or alternative notions of dependence, as the univariate DKW framework does not transfer directly. Developing such tools is an interesting direction, and we plan to investigate it in a future work.

\subsection*{Acknowledgements}
The first author acknowledges the IEMS Alumni Fellowship at Northwestern University for financial support during the conduct of this research. The second and last authors acknowledge support from NIAID grant R01AI168144. The authors also thank the four anonymous reviewers for their useful comments and suggestions, which significantly improved the readability of the paper.

\newpage
\bibliographystyle{apalike}    
\bibliography{biblio,concentration_91}

\clearpage
\appendix
\thispagestyle{empty}

\onecolumn
\aistatstitle{Supplementary Materials for ``Nonparametric Multi Change Point Detection for Markov Chains via Adaptive Clustering'' \\
}

\section{Proofs}

\subsection{Proof of Theorem \ref{thm:marginal}}\label{sec:prf-marginal}

    This section is dedicated to the proof of Theorem \ref{thm:marginal}. As detailed in the sketch of the proof, Theorem \ref{thm:marginal} will follow as a consequence of a series of 6 lemmas (Lemma \ref{lemma:KLD-bound}-\ref{lemma:BIC-exactrecovery}). In this section, we only prove Lemmas \ref{lemma:BIC-minimalrecovery} and \ref{lemma:BIC-exactrecovery}, while the proofs of Lemmas \ref{lemma:KLD-bound}-\ref{lemma:monotone3} are not directly related and thus deferred till later in this section. Before presenting our results, we introduce the $\Ocal(\cdot,\cdot,
    \cdot)$ notation for convenience.

    \begin{definition}
     A sequence of random variables $Z_n$ is said to be $\Ocal_p(a_n,b_n,\Cd)$ if 
    \begin{align*}
       \lim_{n\rightarrow\infty}  b_n\prob(|Z_n|>\Cd a_n)\leq \epsilon\numberthis\label{eq:OcalN}
    \end{align*}
    where $a_n,b_n$ are sequence of positive real numbers, and $\Cd>0$ is assumed to be independent of $n$.
    \end{definition}

    We now state Lemmas \ref{lemma:KLD-bound}-\ref{lemma:monotone2} starting with Lemma \ref{lemma:KLD-bound} provides a deviation bound for the empirical distribution in a neighborhood of the change point.  It is proved in \S \ref{sec:prf_KLDbd}.
    \begin{lemma}\label{lemma:KLD-bound}
        Let $l<k$ be two time points and let $n_{kl}:=l-k$, and
        \begin{align*}
            \xi_m(k,l):={n_{kl}}\int_{X_{(1)}}^{X_{(n)}}\lp \hat F_k^l(u) - F_m(u) \rp^2 d\hat F_n(u),\numberthis\label{eq:xi}
        \end{align*}
        and assume that the conditions in Assumptions \ref{assume:consistency}, and \ref{assume:regenerate} hold.
        Then, with $\delta_n$ as specified in the hypothesis of Theorem \ref{thm:marginal}, $\kappa^\star$, $\kappa$ are as in the statement of Theorem \ref{prop:talagrand-scaled}, $\kappa_n^\dagger$  be a large enough positive constant such that $\kappa_i(\rho,\lambda)\exp(-\kappa_i(\rho,\lambda)\kappa_n^\dagger\log x)<\epsilon/x$, for all $i\in 1,\dots, K_n+1$, and $\epsilon$ is as in \cref{eq:OcalN}, we have 
        \begin{align*}
            \sup_{\tau_{m-1}\leq k<l\leq\tau_{m}}  \xi_m(k,l)=\Ocal_p(u_n,K_n,\Cd), \text{ where }\ u_n:=8{\log(K_n \delta_n^2)\log \delta_n}.\numberthis\label{eq:un}
        \end{align*}
    \end{lemma}
 Before stating the next result, we introduce $S_n(\tau_1',\dots,\tau_L')$ as the sum of the between time points $(\tau_1,\dots,\tau_L)$. Formally, 
    \begin{align*}
        S_n(\tau_1',\tau_2',\dots,\tau_L'):=(\tau_{i+1}'-\tau_i')\sum_{i=1}^{L-1}\int_{X_{(1)}}^{X_{(n)}} \hat F_{\tau_i'}^{\tau_{i+1}'}(u)\lp1-\hat F_{\tau_i'}^{\tau_{i+1}'}(u)\rp d\hat F_n(u)\numberthis\label{eq:def_EAUC}
    \end{align*}     
    Next, we require another Lemma which establishes the monotonicity property of the risk and is proved in \S \ref{sec:prf-monotne}.

    \begin{lemma}\label{lemma:monotne}
        Under Assumptions \ref{assume:consistency} and \ref{assume:regenerate}, and for any integer $s\in \{1,\dots,K_n\}$ and for all $L\geq 1$, such that $\tau_s<\tau_1'<\tau_2'<\dots<\tau_L'<\tau_{s+1}$,
        we have
        \begin{align*}
            0\geq R_n(\tau_s,\tau_1',\dots,\tau_L',\tau_{s+1})-R_n(\tau_s,\tau_{s+1}) = \Ocal_p(u_n\pow L,K_n,\Cd)
        \end{align*}
    where $u_n\pow L=8\Cd L{{\log(LK_n \delta_n^2)\log \delta_n}}$, and $\Cd$ is as in Lemma \ref{lemma:KLD-bound}.
    \end{lemma}
    
    The following lemma (proved in \S \ref{sec:prf-mntne3}) is a generalisation of the previous lemma.
    \begin{lemma}\label{lemma:monotone3}
        Under Assumptions \ref{assume:consistency} and \ref{assume:regenerate}, let $\tilde\tau_1<\dots<\tilde\tau_s$ be a collection of $s$ time points and let $\tau_1'<\dots<\tau_p'$ be a collection of $p$ time points. Then,
            \begin{align*}
                R_n(\tilde\tau_1,\dots,\tilde\tau_s)\geq R_n(\tilde\tau_1,\dots,\tilde\tau_s,\tau_1',\dots,\tau_p')
            \end{align*}
            where the risk for an unordered sequence of time points $\tau_1,\dots,\tau_s,\tau_1',\dots,\tau_p'$ is defined as in \cref{eq:risk_unordered}.
    \end{lemma}

    Since
    \begin{align*}
        S_n(\tau_s,\tau_1',\dots,\tau_L',\tau_{s+1})-S_n(\tau_s,\tau_{s+1})=R_n(\tau_s,\tau_1',\dots,\tau_L',\tau_{s+1})-R_n(\tau_s,\tau_{s+1})
    \end{align*}
    the following Lemma can be proved as a corollary to Lemma \ref{lemma:monotne}.
    
    \begin{lemma}\label{lemma:monotone2} 
        Under Assumptions \ref{assume:consistency} and \ref{assume:regenerate}, and for any integer $s\in \{1,\dots,K_n\}$ and for all $L\geq 1$, such that $\tau_s<\tau_1'<\tau_2'<\dots<\tau_L'<\tau_{s+1}$,
        we have
        \begin{align*}
            0\geq S_n(\tau_s,\tau_1',\dots,\tau_L',\tau_{s+1})-S_n(\tau_s,\tau_{s+1}) = \Ocal_p(u_n\pow L,K_n,\Cd)
        \end{align*}
        where $u_n\pow L=8\Cd L{{\log(LK_n \delta_n^2)\log \delta_n}}$, and $\Cd$ is as in Lemma \ref{lemma:KLD-bound}.
    \end{lemma}
    
    We move on to stating our next result which proves that under Assumptions \ref{assume:consistency}, \ref{assume:identifiability}, and \ref{assume:regenerate}, we detect at least the correct number of change points.

    \begin{lemma}\label{lemma:BIC-minimalrecovery}
        Under Assumptions \ref{assume:consistency}, \ref{assume:identifiability}, and \ref{assume:regenerate}, and the hypothesis of Theorem \ref{thm:marginal}, we have $\prob(\hat K_n\geq K_n)\rightarrow 1$.
    \end{lemma}
    \subsection{Proof of Lemma \ref{lemma:BIC-minimalrecovery}}\label{sec:prf-BICminimal}
    \begin{proof}
    The proof of this lemma is divided into two steps. We will first prove a deviation bound of $\operatorname{BIC_{K_n}}$ from $\operatorname{BIC_L}$. Then the rest of the proof follows using an union bound on the events of incurring an error for every $1<L<K_n$. We begin with the first step.
    
    \textbf{Step I:} We will first prove the following statement via induction: For any $L<K_n$
    \begin{align*}
            \mathrm{BIC_L}-\mathrm{BIC_{K_n}}\geq 3(K_n-L)\eta_{\min{}}-(K_n-L)(K_n+5)\Ocal_p(u_n\pow{K_n},K_n,\Cd)-(K_n-L)\zeta_n.\numberthis\label{eq:bicinddiff}
        \end{align*}
        For $r=1,\dots,K_n$ let $\Bcal_r(L,\delta_n)$ be the set of all estimated change points for which at least one $\tau_r$ is $\delta_n$ away from every estimate. Formally, we define 
        \begin{align*}
            \Bcal_r(L,\delta_n) & := \lc (\tau_1',\dots,\tau_L'): \tau_1'<\dots<\tau_L' \text{ and } |\tau_s'-\tau_r|>\delta_n \forall\ 1\leq s\leq L \rc.\numberthis\label{eq:estimation_ball}
        \end{align*}
        For $L = K_n-1$, by pigeon-hole principle, there exists at least one $\tau_r$ such that $|\tau_i'-\tau_r|$ is large for all $i$. Therefore, the estimated change points $(\hat \tau_1,\dots,\hat \tau_L)\in \Bcal_r(L,\delta_n)$ for some $r$. 
        
        Let $(\tau_1',\dots,\tau_L')$ be any element of $\Bcal_r(L,\delta_n)$. Let $\tau_{K_n+1}=n$ and for $i=0,\dots,r-1,r+2,\dots,K_n$ let  $\{\tau_{i,1}',\dots,\tau_{i,\max{}}'\}$ be the largest ordered subset of $\{\tau_1',\dots, \tau_L'\}$ with all values between $\tau_{i}$ and $\tau_{i+1}$. If the set $\{\tau_{i,1}',\dots,\tau_{i,\max{}}'\}$ is empty, we trivially define $\tau_{i,1}'$ as $\tau_{i+1}$ and $\tau_{i,\max{}}'$ as $\tau_{i-1}$. We define
        \begin{align*}
            T_i & := S_n(\tau_{i},\tau_{i,1}',\dots,\tau_{i,\max{}}',\tau_{i+1})\\
            T_r & :=  S_n(\tau_{r-1},\tau_{r-\delta_n})\\
            T_{r+1} & := S_n(\tau_{r+\delta_n},\tau_{r+1})
        \end{align*}
         and $T_{K_{n}+2}$ as defined below in \cref{eq:TK-decomp}. Now, using Lemma \ref{lemma:monotone3}, we have
        \begin{equation}\label{triangle_Rn}
        \begin{aligned}
          R_n(\tau_1',\dots,\tau_L') & \geq R_n(\tau_1',\dots,\tau_L',\tau_1,\dots,\tau_{r-1},\tau_{r-\delta_n},\tau_{r+\delta_n},\tau_{r+1},\dots,\tau_{K_n})\\
            & = T_0+T_1+\dots+T_{K_n+2}
        \end{aligned}
       \end{equation}
with the last equality following from the definition of $T_i$'s.
        It follows using Lemma \ref{lemma:monotone2} that for all $i=0,\dots,r-1,r+2,\dots,K_n$
        \begin{align*}
           S_n(\tau_i,\tau_{i+1}) \geq T_{{i}}\geq S_n(\tau_i,\tau_{i+1})+\Ocal_p(u_n\pow{K_n},K_n,\Cd). 
        \end{align*}

        It follows by trivially subtracting $\Ocal_p$ terms that,
        \begin{align*}
            T_r & \geq S_n(\tau_{r-1},\tau_{r-\delta_n})+\Ocal_p(u_n\pow{K_n},K_n,\Cd)\\
            T_{r+1} & \geq  S_n(\tau_{r+\delta_n},\tau_{r+1})+\Ocal_p(u_n\pow{K_n},K_n,\Cd).
        \end{align*}
           Finally, we are left with $T_{K_n+2}$ which is 
           \begin{align}
        T_{K_{n+2}}
        &= S_n(\tau_r- \delta_n,\tau_r+ \delta_n)+S_n(\tau_r+ \delta_n,\tau_{r+1}) \notag\\
        &= S_n(\tau_r- \delta_n,\tau_r)+S_n(\tau_r,\tau_r+ \delta_n)+\Delta_S, \label{eq:TK-decomp}
        \end{align}
        where (with $F_{r-1,1/2}$ as defined in \cref{eq:frhalf})
        \begin{equation}
            F_{r-1,1/2}(u) = \frac{F_{r-1}(u)+F_r(u)}{2}.\label{eq:frhalf}
        \end{equation}
        \begin{align*}
        \Delta_S
        & := S_n(\tau_r- \delta_n,\tau_r+ \delta_n)
              -S_n(\tau_r- \delta_n,\tau_r)
              -S_n(\tau_r,\tau_r+ \delta_n) \notag\\
        & = 2 \delta_n \int_{X_{(1)}}^{X_{(n)}} 
              \hat F_{\tau_r- \delta_n}^{\,\tau_r+ \delta_n}(u)
              \Bigl(1-\hat F_{\tau_r- \delta_n}^{\,\tau_r+ \delta_n}(u)\Bigr)\,
              d\hat F_n(u) \notag\\
        &\quad- \delta_n \int_{X_{(1)}}^{X_{(n)}} 
              \hat F_{\tau_r- \delta_n}^{\,\tau_r}(u)
              \Bigl(1-\hat F_{\tau_r- \delta_n}^{\,\tau_r}(u)\Bigr)\,
              d\hat F_n(u) \notag\\
        &\quad- \delta_n \int_{X_{(1)}}^{X_{(n)}} 
              \hat F_{\tau_r}^{\,\tau_r+ \delta_n}(u)
              \Bigl(1-\hat F_{\tau_r}^{\,\tau_r+ \delta_n}(u)\Bigr)\, d\hat F_n(u) \numberthis\label{eq:DeltaS}\\
        & = -  \int_{X_{(1)}}^{X_{(n)}} \lb
              2 \delta_n(\hat F_{\tau_r- \delta_n}^{\,\tau_r+ \delta_n}(u))^2 -  
               \delta_n(\hat F_{\tau_r- \delta_n}^{\,\tau_r}(u))^2
               -  
               \delta_n(\hat F_{\tau_r}^{\,\tau_r+ \delta_n}(u))^2\, \rb d\hat F_n(u)\\
        & =-  \int_{X_{(1)}}^{X_{(n)}} \lb
              2 \delta_n(\hat F_{\tau_r- \delta_n}^{\,\tau_r+ \delta_n}(u))^2 -  
               \delta_n(\hat F_{\tau_r- \delta_n}^{\,\tau_r}(u))^2 - \delta_n(\hat F_{\tau_r}^{\,\tau_r+ \delta_n}(u))^2\, +\Hcal-\Hcal \rb d\hat F_n(u)\numberthis\label{eq:signalequation1}
\end{align*}
where (with the dependence on $u$ implicit for convenience)
\begin{align*}
    \Hcal & =  - 4\delta_n\hat F_{\tau_r- \delta_n}^{\,\tau_r+ \delta_n} F_{r-1,1/2}+2\delta_n (F_{r-1,1/2})^2  \\
    & \quad + 2\delta_n\hat F_{\tau_r- \delta_n}^{\,\tau_r} F_{r-1}-\delta_n (F_{r-1})^2\\
    & \quad + 2\delta_n\hat F_{\tau_r}^{\,\tau_r + \delta_n} F_{r}-\delta_n (F_{r})^2\\
    & = \delta_n \lp \hat F_{\tau_r- \delta_n}^{\,\tau_r}-\hat F_{\tau_r}^{\,\tau_r+ \delta_n}\rp\lp F_{r-1}-F_r \rp\\
    & \xrightarrow{n\rightarrow\infty}\delta_n \lp F_{r-1}-F_r \rp^2+\ocal(\delta_n)\label{eq:signalequation2}\numberthis
\end{align*}
uniformly in $u$.

Now using Lemma \ref{sec:prf_KLDbd}, the right hand side of \cref{eq:signalequation1} becomes
\begin{align*}
   &  - \! \int_{X_{(1)}}^{X_{(n)}}\! \lb
              2 \delta_n(\hat F_{\tau_r- \delta_n}^{\,\tau_r+ \delta_n}(u)\!-\!F_{r-1,1/2}(u))^2\!-\!  
               \delta_n(\hat F_{\tau_r- \delta_n}^{\,\tau_r}(u)\!-\!F_{r-1}(u))^2\!-\! \delta_n(\hat F_{\tau_r}^{\,\tau_r+ \delta_n}(u)\!-\!F_r(u))^2\!-\!\Hcal\rb\! d\hat F_n\!(u)\!\\
       & =3\Ocal_p(u_n\pow{K_n},K_n,\Cd)-\int_{X_{(1)}}^{X_{(n)}} \Hcal d\hat F_n(u).
\end{align*}
Observe that $d\hat F_n(u) = 0$ for all $u\notin [X_{(1)},X_{(n)}]$ which, using implies Assumption \ref{assume:consistency}, and \cref{eq:signalequation2}
\begin{align*}
    \int_{X_{(1)}}^{X_{(n)}} \Hcal d\hat F_n & = \int_{0}^{1} \Hcal d\hat F_n\\
    & \xrightarrow{n\rightarrow\infty} \delta_n\int_{0}^{1}  \lp F_{r-1}(u)-F_r(u) \rp^2 dF(u)+\ocal_p(\delta_n)\\
    & = \delta_n \int \eta_r(u)dF(u) +\ocal_p(\delta_n).
\end{align*}       

        Now it follows by substituting the previous bounds in \cref{triangle_Rn} that,
        \begin{align*}
            & \min_{(\tau_1',\dots,\tau_L')\in B_r(L,\delta_n)} R_n(\tau_1',\dots,\tau_L') \\
            & \quad \geq\min_{(\tau_1',\dots,\tau_L')\in B_r(L,\delta_n)} R_n(\tau_1',\dots,\tau_L',\tau_1,\dots,\tau_{r-1},\tau_{r-\delta_n},\tau_{r+\delta_n},\tau_{r+1},\dots,\tau_{K_n})\\
            & \quad \geq \, R_n(\tau_1,\dots,\tau_{K_n})+{(K_n+5)}\Ocal_p(u_n\pow{K_n},K_n,\Cd)-\delta_n\int \eta_r(u)dF(u)+\ocal_p(\delta_n).
        \end{align*}
         
        Let $\mathrm{BIC_{K_n}}= - R_n(\tau_1,\dots,\tau_{K_n})+K_n\zeta_n$ and recall from \cref{eq:bic} the definition of $\mathrm{BIC_L}$.  
        Following the calculations above, we get 
        \begin{align*}
            \mathrm{BIC_L}-\mathrm{BIC_{K_n}}&\leq -\delta_n\int \eta_r(u)dF(u)-(K_n+5)\Ocal_p(u_n^{(K_n)},K_n,\Cd)-\zeta_n\\
            & \leq -\delta_n\eta_{\min{}}-(K_n+5)\Ocal_p(u_n^{(K_n)},K_n,\Cd)-\zeta_n \numberthis\label{eq:bicdiff}
        \end{align*}
        where the second inequality follows from the definition of $\eta_{\min{}}$ in Assumption \ref{assume:identifiability}. This proves our induction hypothesis for $L=K_n-1$.
        
Now let the hypothesis of the induction (eq. (\ref{eq:bicinddiff})) hold true for $L=K_n-r$. Reparametrizing $K_n-r$ as $K_n'$, we now show for $L=K_n'-1$. It is easy to see invoking \cref{eq:bicdiff} that,
        \begin{align*}
            \mathrm{BIC_{L}}-\mathrm{BIC_{K_n'}}
            & \leq -\delta_n\eta_{\min{}}-(K_n+5)\Ocal_p(u_n^{(K_n')},K_n',\Cd)-\zeta_n \\
            & \leq -\delta_n\eta_{\min{}}-(K_n+5)\Ocal_p(u_n^{(K_n)},K_n,\Cd)-\zeta_n
        \end{align*}
        where inequality follows since for $K_n'<K_n$, $|\Ocal_p(u_n^{(K_n')},K_n',\Cd)|\leq |\Ocal_p(u_n^{(K_n)},K_n,\Cd)|$.
        Therefore using induction,
        \begin{align*}
            \mathrm{BIC_{L+1}}-\mathrm{BIC_{K_n}} & = \mathrm{BIC_{L+1}}-\mathrm{BIC_L}+\mathrm{BIC_L}-\mathrm{BIC_{K_n}} \\
            & \leq -\delta_n(K_n-r+1)\eta_{\min{}}-K_n(K_n+5)\Ocal_p(u_n^{(K_n)},K_n,\Cd)-(K_n-r+1)\zeta_n
        \end{align*}
        which is what we required.
    
        \textbf{Step II:} It now follows that, for any $n$,
        \begin{align*}
            \prob(\hat K_n<K_n) & =\prob\lp \bigcup_{L=1}^{K_n-1}\lc \mathrm{BIC_L}>\mathrm{BIC_{K_n}}\rc \rp\\
            & \leq \sum_{L=1}^{K_n-1}\prob\lp \mathrm{BIC_L}>\mathrm{BIC_{K_n}} \rp\\
            & {\leq}   \sum_{L=1}^{K_n-1} \prob\lp K_n(K_n+5) \Ocal_p(u_n^{(K_n)},K_n,\Cd)<-\delta_n\eta_{\min{}} - (K_n-L)\zeta_n \rp\\
            & =\sum_{L=1}^{K_n-1} \prob\lp  \Ocal_p(u_n^{(K_n)},K_n,\Cd)<-\frac{\delta_n}{5K_n^2}\eta_{\min{}} - (K_n-L)\frac{\zeta_n}{5K_n^2} \rp\\
            & =\sum_{L=1}^{K_n-1} \prob\lp  \Ocal_p(u_n^{(K_n)},K_n,\Cd)<-\frac{u_n\pow {K_n} \delta_n\log K_n}{40\kappa_n^\dagger K_n^3(\log \delta_n)^2(\log K_n)^2}\eta_{\min{}} - (K_n-L)\frac{\zeta_n}{5K_n^2} \rp
        \end{align*}
        where the last inequality follows by substituting the terms from \cref{eq:bicdiff}.

        Recall that under the hypothesis of the theorem, 
        \[K_n^3\log( \delta_n)^2(\log K_n)^2/\delta_n=\Ocal(1), \text{\quad  so that \quad } \delta_n/K_n^3(\log \delta_n)^2(\log K_n)^2=\Omega(1).\] 
        Therefore, for some large $n$, we have \[\frac{ \delta_n\log K_n}{40\kappa_n^\dagger K_n^3\log( \delta_n)^2(\log K_n)^2}\eta_{\min{}}>1\] This implies

        \begin{align*}
            & \sum_{L=1}^{K_n-1} \prob\lp  \Ocal_p(u_n^{(K_n)},K_n,\Cd)<-\frac{u_n\pow {K_n} \delta_n\log K_n}{5K_n^3\log( \delta_n)^2(\log K_n)^2}\eta_{\min{}} - (K_n-L)\frac{\zeta_n}{5K_n^2} \rp\\
            & \quad \leq \sum_{L=1}^{K_n-1} \prob\lp  \Ocal_p(u_n^{(K_n)},K_n,\Cd)<-u_n\pow{K_n} - (K_n-L)\frac{\zeta_n}{5K_n^2} \rp\numberthis\label{eq:onetamin}\\
            & \quad \leq \sum_{L=1}^{K_n-1} \prob\lp  \Ocal_p(u_n^{(K_n)},K_n,\Cd)<-u_n\pow{K_n} \rp\\
            & \quad = K_n  \prob\lp  \Ocal_p(u_n^{(K_n)},K_n,\Cd)<-u_n\pow{K_n} \rp\\
            & \quad < \epsilon.
        \end{align*}
        This completes the proof.
    \end{proof}
    The next lemma establishes that no more than the correct number of change points is detected.
    \begin{lemma}\label{lemma:BIC-exactrecovery}
            Under Assumptions \ref{assume:consistency}, \ref{assume:identifiability}, and \ref{assume:regenerate}, and the hypothesis of Theorem \ref{thm:marginal}, $\prob\lp \hat K_n>K_n \rp\rightarrow 0$.
    \end{lemma}
    \subsection{Proof of Lemma \ref{lemma:BIC-exactrecovery}}\label{sec:prf-BICexact}
    \begin{proof}
        To avoid trivialities, assume that $K_n>0$ and consider two cases:
        
        \emph{Case I $((\hat \tau_1,\dots,\hat \tau_L)\in \bigcup_{r=1}^{K_n}\Bcal_r(L,\delta_n))$:} When $(\hat \tau_1,\dots,\hat \tau_L)\in \Bcal_r(L,\delta_n)$, for some $r\in\{1,\dots,K_n\}$ (with $\Bcal_r(L,\delta_n)$ defined as in \cref{eq:estimation_ball}),  it follows similarly to the proof of Lemma \ref{lemma:BIC-minimalrecovery} that 
        \begin{align*}
            \prob\lp \bigcup_{L=1}^{K_n-1}\lc \mathrm{BIC_L}>\mathrm{BIC_{K_n}}\rc \rp & \leq K_n\prob\lp \Ocal_p(u_n^{(K_n)},K_n,\Cd)>u_n\pow L\rp\\
            & < \epsilon.
        \end{align*}

        \emph{Case II $((\hat \tau_1,\dots,\hat \tau_L)\in  \Ccal(L,\delta_n))$:} To formalize the second case, we introduce the following notation
        \begin{align*}
             \Ccal(L,\delta_n):=\{(\tau_1',\dots,\tau_{L}'):1<\tau_1'<\dots<\tau_{L}'\leq n, \text{and}\ \exists\  \text{i such that} \ |\tau_i'-\tau_s|\leq \delta_n \ \forall \  1\leq s\leq L \}
        \end{align*}
         We will show that
        \begin{align*}
            R_n(\tau_1',\dots,\tau_L')\geq R_n(\tau_1,\dots,\tau_{K_n})+4L\Ocal_p(u_n\pow{K_n},K_n,\Cd).\numberthis\label{eq:loss-gain}
        \end{align*}
        

        Following \cref{triangle_Rn},
        \begin{align*}
          R_n(\tau_1',\dots,\tau_L') & \geq R_n(\tau_1',\dots,\tau_L',\tau_1,\dots,\tau_{r-1},\tau_r,\tau_{r+1},\dots,\tau_{K_n})\\
            & = T_0+T_1+\dots+T_{K_n}\\
            & \geq \sum_{i=0}^{K_n} S_n(\tau_i,\tau_{i+1}) + K_n \Ocal_p(u_n\pow{K_n},K_n,\Cd)\\
            & = R_n(\tau_1,\dots,\tau_n) + K_n \Ocal_p(u_n\pow{K_n},K_n,\Cd).
        \end{align*}
        Therefore,
        \begin{align*}
             R_n(\tau_1',\dots,\tau_L') & \geq R_n(\tau_1,\dots,\tau_{K_n})+K_n\Ocal_p(u_n,K_n,\Cd).
        \end{align*}
        
        As before, let $\operatorname{BIC_{K_n}}=\operatorname{BIC}_\star=- R_n(\tau_1,\dots,\tau_{K_n})+K_n\zeta_n$. Then, for any given $L$,
        \begin{align*}
            \mathrm{BIC_L}-\mathrm{BIC_{K_n}}\leq -K_n\Ocal_p(u_n,K_n,\Cd)-(K_n-L)\zeta_n
        \end{align*}
        Therefore, 
        \begin{align*}
            \prob\lp \hat K_n > K_n \rp & = \prob\lp \bigcup_{L=K_n+1}^{\bar K} \lc \operatorname{BIC_L}>\operatorname{BIC_{K_n}} \rc \rp\\
            & = \sum_{L=K_n+1}^{\bar K}\prob\lp \operatorname{BIC_L}>\operatorname{BIC_{K_n}} \rp\\
            & \leq \sum_{L=K_n+1}^{\bar K}\prob\lp 4K_n\Ocal_p(u_n^{(K_n)},K_n,\Cd) > (L-K_n)\zeta_n  \rp\\
            & \leq \sum_{L=K_n+1}^{\bar K}\prob\lp \Ocal_p(u_n^{(K_n)},K_n,\Cd) > \lp \frac{L}{K_n}-1\rp\frac{\zeta_n}{4}  \rp\\
            & \leq \bar K \prob\lp \Ocal_p(u_n^{(K_n)},K_n,\Cd) > \lp \frac{1}{K_n}\rp\frac{\zeta_n}{4}  \rp\\
            & \overset{(i)}{\leq} \bar K \prob\lp \Ocal_p(u_n^{(K_n)},K_n,\Cd) > u_n\pow{K_n} \rp\\
            & = \frac{\bar K}{K_n} K_n\prob\lp \Ocal_p(u_n^{(K_n)},K_n,\Cd) > u_n\pow{K_n} \rp\\
            & \overset{(i)}{<} \frac{\bar K}{K_n} \epsilon
        \end{align*}
        where $(i)$ follows from \cref{eq:OcalN}. Recall that $\bar K$ is bounded and we have assumed $K_n$ to be positive. Since $\epsilon$ is arbitrary, the proof follows.
    \end{proof}
    
\begin{proof}[Proof of Theorem \ref{thm:marginal}]
         Now we can finally prove our main theorem. It follows by combining Lemma \ref{lemma:BIC-minimalrecovery} and \ref{lemma:BIC-exactrecovery} that $\prob\lp\hat K_n=K_n\rp\rightarrow1$.
\end{proof}

    \subsection{Proof of Lemma \ref{lemma:KLD-bound}}\label{sec:prf_KLDbd}
    \begin{proof}
        We first show that under the hypothesis of the Lemma \ref{lemma:KLD-bound},
        \begin{align*}
             \lim_{n\rightarrow\infty} K_n\prob \lp\sup_{\tau_{m-1}\leq k<l\leq\tau_{m-1}+\delta_n}  \xi_m(k,l)\geq u_n \rp < \epsilon.\numberthis\label{eq:KLD-boundsmall}
        \end{align*}
        
        As before, let $n_{kl}=l-k$, $\|\hat F_k^l-F \|_0^1:=\sup_{u\in[0,1]}|\hat F_k^l(u)-F(u)|$ and $E_n$ to be the following event
        \begin{align*}
            E_n:= \bigcup_{k,l}\lc \sqrt{n_{kl}}\|\hat F_k^l-F\|_0^1>\sqrt{{\Cd}{\log(K_n\delta_n^2)\max\lc \log 2,\log n_{kl}\rc}} \rc.
        \end{align*}

        Then using union bound, 
        \begin{align*}
            \prob\lp E_n \rp & \leq \sum_{k,l}\prob \lp \sqrt{n_{kl}} \|\hat F_k^l-F\|_0^1>\sqrt{{\Cd}{\log(K_n\delta_n^2)\max\lc \log 2,\log n_{kl}\rc}}  \rp\\
            & \overset{(i)}{\leq} \sum_{k,l} \kappa^\star(\rho,
        \lambda)\exp\lp -\kappa(\rho,\lambda) \Cd\frac{n_{kl}}{\log n_{kl}} \frac{\log(K_n\delta_n^2)\max\lc \log 2,\log n_{kl}\rc}{n_{kl}} \rp\\
        &\overset{(ii)}{\leq} \delta_n^2 \kappa^\star(\rho,
        \lambda)\exp\lp -\kappa(\rho,\lambda)\Cd  \frac{\log(K_n\delta_n^2)\log n_{kl}}{\log n_{kl}} \rp\\
        &\overset{(iii)}{\leq }\delta_n^2 \kappa^\star(\rho,
        \lambda)\exp\lp -\kappa(\rho,\lambda) \Cd {\log(K_n\delta_n^2)} \rp\\
        & \overset{(iv)}{\leq} \frac{\epsilon\delta_n^2}{K_n\delta_n^2}\\
        & = \frac{\epsilon}{K_n}\numberthis\label{eq:KLD-bounddone}
        \end{align*}
$(i)$ follows from Theorem \ref{prop:talagrand-scaled}; 
$(ii)$ follows because to avoid trivialities we can assume $n_{kl}\geq2$ 
and because $|\{\tau_{m-1}\leq k,l\leq \tau_{m-1}+\delta_n\}|\leq\delta_n^2 $; 
{$(iii)$ follows simplifying the fraction;}
{$(iv)$ }follows from the definition of $\Cd$ in the statement of Lemma \ref{lemma:KLD-bound}.

Now, observe that {for any $m$, $k$ and $l$,}
\begin{align*}
    \xi_m(k,l)= {n_{kl}}\int_{X_{(1)}}^{X_{(n)}}\lv\hat F_k^l(u)-F(u)\rv^2 d\hat F_n(u)\leq {n_{kl}}(\|\hat F_k^l-F\|_0^1)^2.
\end{align*}
    Since $\log \delta_n\geq \log n_{kl}$, $E_n\subseteq\{\sup_{\tau_{m-1}\leq k<l\leq\tau_{m-1}+\delta_n}  \xi_m(k,l)\geq u_n\}$ and it follows from \cref{eq:KLD-bounddone}, that 
\begin{align*}
    K_n\prob\lp \sup_{\tau_{m-1}\leq k<l\leq\tau_{m-1}+\delta_n}  \xi_m(k,l)\geq u_n \rp\leq \epsilon.
\end{align*}
The conclusion of the lemma now follows. This completes the proof.
\end{proof}
\subsection{Proof of Lemma \ref{lemma:monotne}}\label{sec:prf-monotne}
    \begin{proof}
        For convenience of notation, let $n_1^\star=\tau_1'-\tau_s$, and so on until we have $n_L^\star = \tau_L'-\tau_{L-1}'$, and $n_{L+1}^\star = \tau_{s+1}-\tau_L'$. Let $n^\star = \sum_i n_i^\star$. Then, overloading the notation of $\tau_s$ as $\tau_0'$ and $\tau_{s+1}$ as $\tau_{L+1}'$, we observe that 
        \begin{small}
             \begin{align*}
            & R_n(\tau_s,\tau_1',\dots,\tau_L',\tau_{s+1})-R_n(\tau_s,\tau_{s+1})\\ & =\int_{X_{(1)}}^{X_{(n)}}\lb \sum_{i=1}^{L+1} n_i^\star \hat F_{\tau_i'}^{\tau_{i+1}'}(u)\lp 1-\hat F_{\tau_i'}^{\tau_{i+1}'}(u)\rp -n^\star\hat F_{\tau_s}^{\tau_{s+1}}(u)\lp 1-\hat F_{\tau_s}^{\tau_{s+1}}(u)\rp \rb d\hat F_n(u)\\
            \end{align*}
            Looking at the integrand (and dropping $u$ for convenience) we get
            \begin{align*}
            \sum_{i=1}^{L+1} n_i^\star \hat F_{\tau_i'}^{\tau_{i+1}'}\lp 1-\hat F_{\tau_i'}^{\tau_{i+1}'}\rp -n^\star\hat F_{\tau_s}^{\tau_{s+1}}\lp 1-\hat F_{\tau_s}^{\tau_{s+1}}\rp & \overset{(i)}{=} - \sum_{i=1}^{L+1} n_i^\star (\hat F_{\tau_i'}^{\tau_{i+1}'})^2 + n^\star(\hat F_{\tau_s}^{\tau_{s+1}})^2\numberthis\label{eq:newmonotone-eq1}\\
            & \overset{(ii)}{\leq} 0,
        \end{align*}
        \end{small}
       
        where $(i)$ follows since $\sum_{i=1}^{L+1} n_i^\star \hat F_{\tau_i'}^{\tau_{i+1}'}=n^\star F_{\tau_s}^{\tau_{s+1}}$ and $(ii)$ follows since $-x^2$ is a concave function. We turn to proving the order term.

        We have already pointed out the fact that $\sum_{i=1}^{L+1} n_i^\star \hat F_{\tau_i'}^{\tau_{i+1}'}=n^\star F_{\tau_s}^{\tau_{s+1}}$ and recall that $\sum_{i=1}^{L+1} n_i^\star = n^\star$ by definition. Now consider the right hand side of \cref{eq:newmonotone-eq1} and add and subtract $n^\star (F_s)^2-2n^\star \hat F_{\tau_s}^{\tau_{s+1}}F_s$ to get,
        \begin{align*}
            - \sum_{i=1}^{L+1} n_i^\star (\hat F_{\tau_i'}^{\tau_{i+1}'})^2 + n^\star(\hat F_{\tau_s}^{\tau_{s+1}})^2 & =  - \sum_{i=1}^{L+1} n_i^\star\lb (\hat F_{\tau_i'}^{\tau_{i+1}'})^2+ (F_s)^2-2n^\star \hat F_{\tau_i'}^{\tau_{i+1}'}F_s \rb\\
            & \qquad \qquad + n^\star\lb (\hat F_{\tau_s}^{\tau_{s+1}})^2 +(F_s)^2-2 \hat F_{\tau_s}^{\tau_{s+1}}F_s\rb \\
            & =  - \sum_{i=1}^{L+1} n_i^\star (\hat F_{\tau_i'}^{\tau_{i+1}'}-F_s)^2+ \underbrace{n^\star (\hat F_{\tau_s}^{\tau_{s+1}}-F_s)^2}_{\geq 0}\\
            & \geq  - \sum_{i=1}^{L+1} n_i^\star (\hat F_{\tau_i'}^{\tau_{i+1}'}-F_s)^2.
        \end{align*}
        Therefore, 
        \begin{align*}
            R_n(\tau_s,\tau_1',\dots,\tau_L',\tau_{s+1})-R_n(\tau_s,\tau_{s+1}) & \geq \int_{X_{(i)}}^{X_{(n)}} - \sum_{i=1}^{L+1} n_i^\star (\hat F_{\tau_i'}^{\tau_{i+1}'}(u)-F_s(u))^2d\hat F_n(u)\\
            & \geq -\sum_{i=1}^{L+1}\xi_s(\tau_i',\tau_{i+1}').
        \end{align*}
         For any collection of positive random variables $Z_1,\dots,Z_{L+3}$, recall the probabilistic bound
        \begin{equation}\label{sequence_bound}
            \prob\lp \sum_{i=1}^{L+1} Z_i>a \rp\leq \sum_{i=1}^{L+1}\prob\lp Z_i>\frac{a}{L+1} \rp.
        \end{equation}
        Therefore, using Lemma \ref{lemma:KLD-bound}
    \begin{align*}
        & \lim_{n\to\infty} K_n \prob\lp 
          \lv R_n(\tau_s,\tau'_1,\ldots,\tau'_L,\tau_{s+1})
          - R_n(\tau_s,\tau_{s+1})\rv
        > 8\Cd L{{\log(LK_n \delta_n^2)\log \delta_n}} \rp\\
        &\qquad \le \lim_{n\to\infty} K_n \prob\lp
           \sum_{i=1}^{L+1}\xi_s(\tau_i',\tau_{i+1}')
           > 8\Cd L{\log(K_n L\delta_n^2)\log \delta_n}
        \rp \\
        &\qquad \le \lim_{n\to\infty} K_n \prob\lp
           \sum_{i=1}^{L+1}\xi_s(\tau_i',\tau_{i+1}')
           > 4\Cd (L+1){\log(LK_n \delta_n^2)\log \delta_n}
        \rp \\
        &\qquad \le L^{-1} \sum_{i=1}^{L+1} \lim_{n\to\infty} L K_n \prob\lp \xi_s(\tau_i',\tau_{i+1}')
           >4\Cd {\log(LK_n \delta_n^2)\log \delta_n}
        \rp \\
        &\qquad < L^{-1}(L+1)\,\varepsilon \,.
        \end{align*}
        Since $n>\delta_n$, the statement of Lemma \ref{lemma:monotne} is now established.
    \end{proof}

\subsection{Proof of Lemma \ref{lemma:monotone3}}\label{sec:prf-mntne3}   
\begin{proof}
Fix one “coarse’’ interval $J=(\tilde\tau_j,\tilde\tau_{j+1})$ induced by
$\{\tilde\tau_1,\ldots,\tilde\tau_s\}$. The extra points
$\{\tau_1',\ldots,\tau_p'\}$ partition $J$ into subintervals
$J_1,\ldots,J_{m_j}$ with counts $n_{j,k}^\star$ (so that
$\sum_{k=1}^{m_j}n_{j,k}^\star=n_J^\star$) and
$p_{j,k}(u):=\hat F_{J_k}(u)\in[0,1]$. Set weights
$w_{j,k}:=n_{j,k}^\star/n_J^\star$ and the weighted average
$\bar p_j(u):=\sum_{k=1}^{m_j} w_{j,k} p_{j,k}(u)=\hat F_J(u)$.

Let $f(x)=x(1-x)=x-x^2$, a concave function on $[0,1]$. Then, pointwise in $u$,
\begin{align*}
\sum_{k=1}^{m_j} n_{j,k}^\star f \bigl(p_{j,k}(u)\bigr)
- n_J^\star f \bigl(\bar p_j(u)\bigr)
&= -\,n_J^\star \left(
\sum_{k=1}^{m_j} w_{j,k}\,p_{j,k}(u)^2
- \Bigl(\sum_{k=1}^{m_j} w_{j,k}\,p_{j,k}(u)\Bigr)^2
\right)\\
&= -\,n_J^\star\,\Var_{w_j} \bigl(p_{j,k}(u)\bigr) \le 0.
\end{align*}
Equivalently,
\[
\sum_{k=1}^{m_j} n_{j,k}^\star \hat F_{J_k}(u)\bigl(1-\hat F_{J_k}(u)\bigr)
 \le  n_J^\star \hat F_J(u)\bigl(1-\hat F_J(u)\bigr).
\]
Summing this inequality over all coarse intervals
$J\in\mathcal J(\{\tilde\tau_1,\ldots,\tilde\tau_s\})$ and integrating
with respect to $d\hat F_n(u)$ yields
\[
R_n(\tilde\tau_1,\ldots,\tilde\tau_s,\tau_1',\ldots,\tau_p')
 \le  R_n(\tilde\tau_1,\ldots,\tilde\tau_s).
\]
Equality holds iff $\Var_{w_j}(p_{j,k}(u))=0$ for all $j$ and a.e.\ $u$,
i.e., when each coarse interval has identical within-subinterval cdfs.
\end{proof}

\subsection{Proof of Theorem \ref{prop:talagrand-scaled}}\label{sec:prf-tal-scaled}

    Now we proceed to proving Theorem \ref{prop:talagrand-scaled}. Since there is always a rational number between any two real numbers, it holds almost everywhere that
    \begin{align*}
        \sup_{f\in\Fcal_{[0,1]}}\lv\sum_{i=1}^n\lp f(Y_i) - \expec_{\pi}\lb f(Y)\rb\rp \rv\leq \sup_{f\in\Fcal_{[0,1]}\bigcap \Qbb}\lv \sum_{i=1}^n\lp f(Y_i) - \expec_{\pi}\lb f(Y)\rb\rp \rv+2
    \end{align*}
    Therefore,
    \begin{align*}
        \prob(nZ>nt+\kappa \Rcal(\Fcal_{[0,1]\bigcap \Qbb})))\leq \prob\lp \underbrace{\sup_{f\in\Fcal_{[0,1]}\bigcap \Qbb}\lv\sum_{i=1}^n\lp f(Y_i) - \expec_{\pi}\lb f(Y)\rb\rp \rv>nt+\kappa \Rcal(\Fcal_{[0,1]\bigcap \Qbb}))-2}_{=:\Tcal}\rp
    \end{align*} 
    
    $t\geq 3/n$ by hypothesis. Therefore $nt-2\geq 1$. We now state the following Lemma which is proved in \S \ref{sec:prf-talagrand}.
    \begin{lemma}\label{lemma:talagrand}
    Let $Y_1,\dots,Y_n$ be a sequence of random variables from a Markov chain with stationary distribution $\pi$, and let $Y$ be a random variable with distribution $\pi$. Define  
    \begin{align*}
        Z':= \sup_{f\in\Fcal_{[0,1]\bigcap \Qbb}}\lv\sum_{i=1}^n\lp f(Y_i) - \expec_{\pi}\lb f(Y)\rb\rp \rv.
    \end{align*}
    Then, for some universal constant $\kappa>4e$, any $\kappa_\rho>\sqrt{\expec_A[\rho_A^2(2)]}$, $\kappa_\lambda=2\expec_A[\exp(\rho_A(2)\lambda)]/\lambda$, and $\rho_o$ as defined in \S \ref{sec:nummelin-splitting},
   \begin{align*}
    \prob\lp Z'> \rdot&\ldot t+\kappa \Rcal(\Fcal_{[0,1]\bigcap \Qbb}) \rp \leq \exp\lp -\frac{\expec_A[\rho_A(2)]}{\kappa}\min\lc \frac{t^2}{n \expec_A[\rho_A^2(2)]},\frac{t}{\rho_o^3\log n } \rc \rp \text{ and }\\
       \Rcal(\Fcal_{[0,1]\bigcap \Qbb})) & = 2(\expec_A[\rho_A(2)]+\expec_\nu[\rho_A(2)])+\kappa\lb \kappa_\rho\log \lp \frac{\kappa_\rho}{\sqrt{\expec_A[\rho_A^2(2)]}}\rp+\sqrt{n\expec_A[\rho_A^2(2)]\log\lp\frac{\kappa_\rho}{\sqrt{\expec_A[\rho_A^2(2)]}}\rp} \rb\\
        & \qquad +n\exp\lp -\kappa_\rho\lambda/2 \rp\kappa_\lambda.
    \end{align*} 
    \whiteqed
\end{lemma}

It now follows using Lemma \ref{lemma:talagrand} that the right hand side of the previous equation is bounded above by
    \begin{align*}
        \exp\lp -\frac{\expec_A[\rho_A(2)]}{\kappa}\min\lc \frac{(nt-2)^2}{n \expec_A[\rho^2_A(2)]},\frac{(nt-2)}{\rho_o^3\log n } \rc \rp
    \end{align*}
    By setting $\kappa_\rho=(2/\lambda)\log(n/2\kappa_\lambda)$ and observing that under Assumption \ref{assume:regenerate} (EM), $1/(2\kappa_\lambda)< 1$ we get
    \begin{small}
    \begin{align*}
        \Rcal(\Fcal_{[0,1]\bigcap \Qbb}))\leq 2(\expec_A[\rho_A(2)]+\expec_\nu[\rho_A(2)])+\kappa\lb 2\frac{\log(n)}{\lambda}\log \lp \frac{2\log(n)/\lambda}{\sqrt{\expec_A[\rho^2_A(2)]}}\rp+\sqrt{n\expec_A[\rho^2_A(2)]\log\lp\frac{2\log(n)/\lambda}{\sqrt{\expec_A[\rho^2_A(2)]}}\rp} \rb
    \end{align*}
    \end{small}

    Observe that $\sqrt{\expec_A[\rho^2_A(2)]}\geq 1$. Now, with a constant $\kappa(\tau,\lambda)$ depending on $\lambda$ and $\expec_A[\rho^2_A(2)]$, and $\expec_v[\rho_A(2)]$ we have with some standard manipulations 
    \begin{align*}
        \Rcal(\Fcal_{[0,1]\bigcap \Qbb}))\leq \kappa(\rho,\lambda) \sqrt{n\log n} 
    \end{align*}
    Then
    we have 
    \begin{align*}
         \prob(nZ>nt+ \kappa(\rho,\lambda) \sqrt{n\log n}))\leq \exp\lp -\frac{\expec_A[\rho_A(2)]}{\kappa}\min\lc \frac{(nt-2)^2}{n \expec_A[\rho^2_A(2)]},\frac{(nt-2)}{\rho_o^3\log n } \rc \rp.
    \end{align*}
    
    Now dividing both sides of $\Tcal$ by $n$ and trivially upper bounding $2$ by $2\kappa$, we have for some universal constant $\kappa>0$, and for all $t>3/n$
    \begin{align*}
    \prob\lp Z>t+ \kappa(\rho,\lambda) \sqrt{n\log n}/n \rp\leq\exp\lp -\frac{\expec_A[\rho_A(2)]}{\kappa}\min\lc \frac{(nt-2)^2}{n \expec_A[\rho^2_A(2)]},\frac{(nt-2)}{\rho_o^3\log n } \rc \rp\numberthis\label{eq:talagrandscaled-conc}
    \end{align*}
    where, for some constant $\kappa(\tau,\lambda)$ depending only on $\expec_A[\rho^2_A(2)]$, $\expec_v[\rho_A(2)]$, $\lambda$. 
    Next, observe that 
\begin{align*}
    \prob\lp Z>t \rp & = \prob\lp Z-\expec Z > t - \expec Z \rp
\end{align*}
Since $\expec Z<\kappa(\rho,\lambda) \sqrt{n\log n}/n=\Ocal(\sqrt{\log n/n})$, there exists a constant $\kappa'(\rho,\lambda)>3/n$ such that 
\begin{align*}
    t-\expec Z >t -\kappa'(\rho,\lambda).
\end{align*}
Then, 
\begin{align*}
    \prob(Z>t)\leq\exp\lp -\frac{\expec_A[\rho_A(2)]}{\kappa}\min\lc \frac{(n(t-\kappa'(\rho,\lambda))-2)^2}{n \expec_A[\rho^2_A(2)]},\frac{(n(t-\kappa'(\rho,\lambda))-2)}{\rho_o^3\log n } \rc \rp.
\end{align*}
We now make 2 cases.

\textbf{Case I:} When $t>2\kappa'(\rho,\lambda)$, we also have $t-\kappa'(\rho,\lambda)-2/n>t/2$, and hence
\begin{align*}
    \prob(Z>t)\leq \exp\lp -\frac{\expec_A[\rho_A(2)]}{\kappa}\min\lc \frac{(nt/2)^2}{n \expec_A[\rho^2_A(2)]},\frac{nt/2}{\rho_o^3\log n } \rc \rp.
\end{align*}

\textbf{Case II:} When $0<t\leq 2\kappa'(\rho,\lambda)$, there exists a large enough constant $\kappa^\star(\rho,\lambda)$ such that
\begin{align*}
    \prob(Z>t)\leq \kappa^\star(\rho,\lambda)\exp\lp -\frac{\expec_A[\rho_A(2)]}{\kappa}\min\lc \frac{(nt/2)^2}{n \expec_A[\rho^2_A(2)]},\frac{nt/2}{\rho_o^3\log n } \rc \rp.
\end{align*}
It therefore follows that, for some large enough constant $\kappa(\rho,\lambda)$ and for all $t>0$
\begin{align*}
    \prob(Z>t)\leq \kappa^\star(\rho,\lambda)\exp\lp -\frac{\expec_A[\rho_A(2)]}{\kappa}\min\lc \frac{(nt/2)^2}{n \expec_A[\rho^2_A(2)]},\frac{nt/2}{\rho_o^3\log n } \rc \rp.
\end{align*}
Consequently,
\begin{align*}
    \prob(Z>t)& \leq \kappa^\star(\rho,\lambda)\exp\lp -\frac{\expec_A[\rho_A(2)]}{\kappa}\min\lc \frac{(nt/2)^2}{n \expec_A[\rho^2_A(2)]},\frac{nt/2}{\rho_o^3\log n } \rc \rp\\
    & \leq \kappa^\star(\rho,\lambda)\exp\lp -\frac{\expec_A[\rho_A(2)]}{\kappa}\frac{n\min\{t,t^2\}}{4\log n}\min\lc \frac{1}{\expec_A[\rho^2_A(2)]},\frac{1}{\rho_o^3 } \rc \rp.
\end{align*}
Let
\begin{align*}
    \kappa(\rho,\lambda):=\frac{\expec_A[\rho_A(2)]}{4\kappa}\min\lc \frac{1}{\expec_A[\rho^2_A(2)]},\frac{1}{\rho_o^3 } \rc .
\end{align*}
It now follows that, for all $t>0$
\begin{align*}
     \prob(Z>t) \leq \kappa^\star(\rho,\lambda)\exp\lp -\frac{\kappa(\rho,\lambda)n\min\{t,t^2\}}{\log n} \rp.
\end{align*}

\subsection{Proof of Lemma \ref{lemma:talagrand}}\label{sec:prf-talagrand}

Observe from part (ii) of Theorem 4 \cite{bertail_rademacher_2019} that under Assumption \ref{assume:regenerate}, the Rademacher complexity $\Rcal(\Fcal_{[0,1]\bigcap \Qbb}))$ (as defined in definition 7 \cite{bertail_rademacher_2019}) for any class of VC functions with constant envelope $U$ and characteristic $(\kappa_1,v)$ can be upper bounded as
    \begin{align*}
        \Rcal(\Fcal_{[0,1]\bigcap \Qbb})) \leq  \kappa \lb v\kappa_\rho U\log \frac{\kappa_1\kappa_\rho U}{\sigma'}+\sqrt{vn\sigma'\log\frac{\kappa_1\kappa_\rho U}{\sigma'}} \rb+nU\exp(-\kappa_\rho\lambda/2)\kappa_\lambda,\numberthis\label{eq:lemtala-eq1}
    \end{align*}
    where $(\sigma')^2$ is any number such that 
    \begin{align*}
        \sup_{f\in \Fcal_{[0,1]\bigcap \Qbb}}\expec_A\lb\lp\sum_{i=1}^{\rho_A(2)}  f(Y_i)\rp^2\rb\leq (\sigma')^2
    \end{align*}
    and $\kappa_\rho$ is any number such that $0<\sigma'<\kappa_\rho U$.
    To continue with the proof, we first write the following lemma. Its proof is provided in \S \ref{sec:prf-halfinterval} for completeness.
\begin{lemma}\label{lemma:VC-halfinterval}
    $\Fcal_{[0,1]\bigcap \Qbb}$ is VC with constant envelope $1$ and admissible charactaristic $(\kappa,2)$ for some universal constant $\kappa>4e$.
\end{lemma}
    Recall from Lemma \ref{lemma:VC-halfinterval} that the class of all half intervals on rationals $\Fcal_{[0,1]\bigcap \Qbb}$ are VC with a constant envelope $U$ and characteristic $(\kappa,2)$ for some universal constant $\kappa$. Substituting this in \cref{eq:lemtala-eq1}, we get 
    \begin{align*}
        \Rcal(\Fcal_{[0,1]\bigcap \Qbb})) \leq  \kappa \lb 2\kappa_\rho\log \frac{ \kappa_\rho}{\sigma'}+\sqrt{2n\sigma'\log\frac{\kappa_\rho}{\sigma'}} \rb+n\exp(-\kappa_\rho\lambda/2)\kappa_\lambda.
    \end{align*}
    
    Next, we observe that $f(\cdot)$ are indicators of half-intervals. Hence $f(\cdot)\leq 1$ and
    \begin{align*}
        \lp\sum_{i=1}^{\rho_A(2)}  f(Y_i)\rp^2\leq \rho^2_A(2).
    \end{align*}
    Therefore, choosing $(\sigma')^2=\expec_A[\rho^2_A(2)]$ suffices and we get 
    \begin{align*}
        \Rcal(\Fcal_{[0,1]\bigcap \Qbb})) \leq  \kappa \lb 2\kappa_\rho\log \frac{ \kappa_\rho}{\sqrt{\expec_A[\rho^2_A(2)]}}+\sqrt{2n\sqrt{\expec_A[\rho^2_A(2)]}\log\frac{ \kappa_\rho}{\sqrt{\expec_A[\rho^2_A(2)]}}} \rb+n\exp(-\kappa_\rho\lambda/2)\kappa_\lambda.
    \end{align*}
    
    Finally, substituting this into Theorem 5 \cite{bertail_rademacher_2019} and trivially substituting $\log (x\kappa)\leq \kappa\log (x)$ for all large enough constant $\kappa$, we arrive at the required bound
    \begin{align*}
        \Rcal(\Fcal_{[0,1]\bigcap \Qbb})) & = 2(\expec_A[\rho_A(2)]+\expec_\nu[\rho_A(2)])+\kappa\lb \kappa_\rho\log \lp \frac{\kappa_\rho}{\sqrt{\expec_A[\rho^2_A(2)]}}\rp+\sqrt{n\expec_A[\rho^2_A(2)]\log\lp\frac{\kappa_\rho}{\sqrt{\expec_A[\rho^2_A(2)]}}\rp} \rb\\
        & \qquad +n\exp\lp -\kappa_\rho\lambda/2 \rp\kappa_\lambda.
    \end{align*}

    Now, using the exponential tail bound for the suprema of additive functions of regenerative Markov chains (Theorem 6 in \cite{bertail_rademacher_2019}, or Theorem 7 in \cite{adamczak_tail_2008}), we arrive at the conclusion.

\subsection{Proof of Lemma \ref{lemma:VC-halfinterval}}\label{sec:prf-halfinterval}
To prove this lemma, we introduce the notion of VC classes which are commonly used in nonparametric statistics \citep{sen_gentle_2018}. The function $H$ is an envelope for the function class $\mathcal{F}$ with metric $d$ if $|f(x)|\leq
H(x)$ for all $x\in E$ and $f\in\mathcal{F}$. For a metric space
$(\mathcal{F},d)$, the covering number $\mathcal{N}(\epsilon,\mathcal{F},d)$ is
the minimal number of balls of size $\epsilon$ needed to cover $\mathcal{F}
$. The metric that we use here is the
$L_{2}(Q)$-norm denoted by $\Vert.\Vert_{L_{2}(Q)} \ $ and given by $\ \Vert f\Vert_{L_{2}(Q)}=\left(\int{f^{2}dQ}\right)^{1/2}$.

\begin{definition}
A countable class $\mathcal{F}$ of measurable functions on $E\to \mathbb R$ is said to be of VC-type (or
Vapnik-Chervonenkis type) for an envelope $H$ and admissible characteristic $(\kappa,v)$
(positive constants) such that $\kappa\geq(3\sqrt
{e})^{v}$ and $v\geq1$, if for all probability measure $Q$ on $(E,\mathcal{E})$
with $0<\Vert H\Vert_{L_{2}(Q)}<\infty$ and every $0<\epsilon<1$,
\[
\mathcal{N}\left(  \epsilon\Vert H\Vert_{L_{2}(Q)},\,\mathcal{F},\,\Vert
.\Vert_{L_{2}(Q)}\right)  \leq{\kappa}{\epsilon^{-v}}.
\]
\end{definition}

\begin{proof} 
    We provide a proof for completeness. We begin this proof with some requisite definitions. Given a class of indicator functions $\Ical$ defined on $\chi$, and a set $\{x_1,\dots,x_n\}\in\chi^n$, we first define
    \begin{align*}
        \Ical(\{x_1,\dots,x_n\}):=\{ (f(x_1),\dots,f(x_n))\in\{0,1\}^n : f\in\Ical\}
    \end{align*}
    The \textit{growth function} of the $\Fcal$ is then defined as 
    \begin{align*}
        \Delta_n(\Ical) = \max_{\{x_1,\dots,x_n\}\in \chi^n} |\Ical(\{x_1,\dots,x_n\})|
    \end{align*}

The \textit{VC-dimension} of $\Ical$ is then defined as 
\begin{align*}
    VC(\Ical) := \argmax_n\lc n : \Delta_n(\Ical) = 2^n \rc
\end{align*}

We will now show that $VC(\Fcal_{[0,1]\bigcap\Qbb})=1$. Let $\{x_1,\dots,x_n\}$ be any ordered sample. That is, $x_1<x_2<,\dots,<x_n$. For any $t\in [0,1]\bigcap \Qbb$, observe that $(\indicator[x_1<t],\indicator[x_2<t],\dots,\indicator[x_n<t]) 
$ has the form $(1,1,1,\dots,1,0,0\dots,0)$. In particular, the values of $\Fcal_{[0,1]\bigcap\Qbb}(\{x_1,\dots,x_n\})$ has to be within the following set 
\begin{align*}
     & \quad \ (0,0,0,\dots,0),\\
    & \quad \ (1,0,0,\dots,0),\\
    & \quad \ (1,1,0,\dots,0),\\
    & \quad \ \vdots \\
    & \quad \ (1,1,1,\dots,1)
\end{align*}
Therefore $\Delta_n(\Fcal_{[0,1]\bigcap\Qbb}) = {n+1}$. This implies that 
\begin{align*}
 VC(\Fcal_{[0,1]\bigcap\Qbb}) & = \argmax_n \{ \Delta_n(\Fcal_{[0,1]\bigcap\Qbb})=2^n \} \\
 & = \argmax_n\{ n+1=2^n \}\\
 & = 1.
\end{align*}

Now, using standard results of covering number bounds, (Theorem 7.8 of \cite{sen_gentle_2018}, see also Theorem 2.6.4 \cite{van_der_vaart_asymptotic_2000}) we have the following result. For some universal constant $\kappa>0$
\begin{align*}
    \mathcal{N}\left(  \epsilon\Vert H\Vert_{L_{2}(Q)},\,\Fcal_{[0,1]\bigcap \Qbb},\,\Vert
.\Vert_{L_{2}(Q)}\right) & \leq \kappa\times  VC(\Fcal_{[0,1]\bigcap \Qbb})(4e)^{VC(\Fcal_{[0,1]\bigcap \Qbb})}\lp\frac{1}{\epsilon}\rp^{2VC(F_{[0,1]\bigcap \Qbb})}\\
 & \overset{(i)}{\leq}  \frac{\kappa' }{\epsilon^2}.
\end{align*}
where $(i)$ follows by substituting $VC(\Fcal_{[0,1]\bigcap\Qbb})$. This completes the proof. 
\end{proof}
\subsection{Proof of Proposition~\ref{prop:optimization}}\label{sec:prf-optimisation}

We first explain the logical construction leading to the optimization model \eqref{nonlinear_model}. We use the index $i=1,\dots, n$ to represent the time points and $l=1,\dots, L+1$ to represent the segment of the time points (i.e. $l$ represents the segment of time points between $\tau_{l-1}$ and $\tau_l$). 

Instead of determining the positions of change points directly, we introduce binary decision variables to represent whether the time point $i$ belongs to the segment $l$: 
$$z_{i,l}\in\{0,1\}\quad l=1,\dots, L+1,   i=1,\dots, n.$$

Based on the property that a time point can be assigned to only one segment, we have the constraints:
\begin{equation}\label{st_oneassign}
 \sum_{l=1}^{L+1}z_{i,l}=1\quad i=1,\dots, n.   
\end{equation}

Since the length of a segment is assumed to be at least 3, we have the constraints:
\begin{equation}\label{st_length}
   \sum_{i=1}^nz_{i,l}\geq 3\quad l=1,\dots, L+1. 
\end{equation}

Also note that if the time point $i$ is assigned to the segment $l$, then the time point $i+1$ must either remain in the segment $l$ or a later segment $l'>l$. This is represented by the constraints:
\begin{equation}\label{st_mono}
  z_{i,l}\leq \sum_{l'\geq l}z_{i+1,l'}\quad l=1,\dots, L+1,   i=1,\dots, n-1.  
\end{equation}

\begin{proof}
Assume $z_{i,l}^*$ is an optimal solution of \eqref{nonlinear_model}. We now use $z_{i,l}^*$ to represent each part of \eqref{eq:risk} as follows:
$$\sum_{p=\tau_{l-1}'}^{\tau_{l}'}\indicator[X_p\leq u] := \sum_{i=1}^n a_{u,i}z_{i,l}^*$$
$$\tau_{l}'-\tau_{l-1}':=\sum_{i=1}^n z_{i,l}^*$$
$$\hat{F}_{\tau_{l-1}'}^{\tau_{l}'}:=\frac{\sum_{i=1}^n a_{u,i}z_{i,l}^*}{\sum_{i=1}^n z_{i,l}^*}.$$
Then we have:
$$R_n(\tau_1',\dots,\tau_L'):=\sum_{l=1}^{L+1}\sum_{u=1}^{n}  \left(\sum_{i=1}^n a_{i,u}z_{i,l}^*\right)\left(1-\frac{\sum_{i=1}^n a_{i,u}z_{i,l}^*}{\sum_{i=1}^n z_{i,l}^*}\right),$$
where $(\tau_1',\dots,\tau_L')$ are the estimated change points. Therefore, we conclude that $\tau'$ is a solution of  \eqref{eq:risk} if $z^*$ is a solution of \eqref{nonlinear_model} and the objective value of \eqref{nonlinear_model} is the same as \eqref{eq:risk}.

Similarly, if $(\tau_1',\dots,\tau_L')$ are the estimated change points, we can use $\tau_l'-\tau_{l-1}'=\sum_{i=1}^n z_{i,l}^*$ to represent the length of the segments for each $l$. The segmentation implies three properties: 1) each time point can be assigned to one segment; 2) the length of a segment ($\tau_l'-\tau_{l-1}'$) is at least 3 ; and 3) if a time point belongs to a given segment, the next time point cannot be assigned to an early segment. Then if $(\tau_1',\dots,\tau_L')$ are the estimated change points,  $z^*$ should satisfy the constraints \eqref{st_oneassign}, \eqref{st_length} and \eqref{st_mono}, which means that $z^*$ is a solution of \eqref{nonlinear_model}, and \eqref{eq:risk} has the same value as the objective function of \eqref{nonlinear_model}. 
\end{proof}
\subsection{Proof of Proposition~\ref{prop:bilinear}}\label{sec:prf-optimisation_bilinear}
\begin{proof}
Suppose $z^*$ solves \eqref{nonlinear_model}. Following Equation~R1 in \cite{borrero2016simple}, we linearize the fractional terms by defining
\[
k_l^*=\frac{1}{\sum_{i=1}^n z_{i,l}^*}, \qquad \text{so that } \sum_{i=1}^n k_l^* z_{i,l}^*=1.
\]
Then
\begin{equation}\label{t_set_sum}
\begin{aligned}
f_{u,l}^* &= \frac{\sum_{i=1}^n a_{i,u}z_{i,l}^*}{\sum_{i=1}^n z_{i,l}^*}
           = \sum_{i=1}^n a_{i,u} k_l^* z_{i,l}^*, \\
d_{u,l}^* &= 1-f_{u,l}^*,
\end{aligned}
\end{equation}
and hence
\[
s_{u,l}^* = \left(\sum_{i=1}^n a_{i,u}z_{i,l}^*\right)\left(1-\frac{\sum_{i=1}^n a_{i,u}z_{i,l}^*}{\sum_{i=1}^n z_{i,l}^*}\right)
          = \sum_{i=1}^n a_{i,u} z_{i,l}^* d_{u,l}^*.
\]
Thus $(z^*,k^*,f^*,d^*,s^*)$ is feasible for \eqref{linear_model} with the same objective value.

Conversely, if $(z^*,k^*,f^*,d^*,s^*)$ is feasible for \eqref{linear_model}, then $z^*\in\mathcal{Z}$ and satisfies all constraints of \eqref{nonlinear_model}. Moreover,
\[
\sum_{l=1}^{L+1}\sum_{u=1}^n s_{u,l}^* 
  = \sum_{l=1}^{L+1}\sum_{u=1}^n\left(\sum_{i=1}^n a_{i,u}z_{i,l}^*\right)\left(1-\frac{\sum_{i=1}^n a_{i,u}z_{i,l}^*}{\sum_{i=1}^n z_{i,l}^*}\right),
\]
which coincides with the objective of \eqref{nonlinear_model}. Therefore, $z^*$ solves \eqref{nonlinear_model} if and only if $(z^*,k^*,f^*,d^*,s^*)$ solves \eqref{linear_model}, with equal objective values. This completes the proof.
\end{proof}

\end{document}

%% file: commands.tex
\usepackage{times}
\usepackage{tikz}
\usepackage{authblk}
\usepackage[english]{babel}
\usepackage{amssymb}
\usepackage[margin=1in]{geometry}
\usepackage{lineno}
\usepackage{calligra}

\usepackage{amsmath}
\usepackage{amssymb}
\usepackage{graphicx}
\usepackage[colorlinks=true, allcolors=blue]{hyperref}
\usepackage{verbatim}
\usepackage{cleveref}  
\usepackage{enumitem}
\usepackage{bbm}
\usepackage{csquotes}

\numberwithin{equation}{section} 
\newcommand{\lv}{\left|}
\newcommand{\rv}{\right|}

\newcommand{\ldot}{\left.}
\newcommand{\rdot}{\right.}

\newcommand{\Var}{\mathrm{Var}}

\DeclareMathAlphabet{\mathbcal}{U}{dutchcal}{m}{n}

\newcommand{\Cbb}{\mathbb{C}}

\newcommand{\Pbb}{\mathbb{P}}
\newcommand{\Qbb}{\mathbb{Q}}
\newcommand{\Rbb}{\mathbb{R}}

\newcommand{\Bcal}{\mathcal{B}}
\newcommand{\Ccal}{\mathcal{C}}

\newcommand{\Ecal}{\mathcal{E}}
\newcommand{\Fcal}{\mathcal{F}}
\newcommand{\Gcal}{\mathcal{G}}
\newcommand{\Hcal}{\mathcal{H}}
\newcommand{\Ical}{\mathcal{I}}

\newcommand{\Ocal}{\mathcal{O}}

\newcommand{\Scal}{\mathcal{S}}
\newcommand{\Rcal}{\mathcal{R}}
\newcommand{\Tcal}{\mathcal{T}}

\newcommand{\Xcal}{\mathcal{X}}

\newcommand{\ocal}{\mathbcal{o}}

\newcommand{\scal}{\mathbcal{s}}

\newcommand{\constant}{\Cbb}

\newcommand\numberthis{\addtocounter{equation}{1}\tag{\theequation}}

\newcommand{\whiteqed}{\hfill$\square$\par\bigskip}
\newcommand{\prob}{\Pbb}
\newcommand{\expec}{\mathbb{E}}
\newcommand{\probl}{\Pbb}

\newcommand{\indicator}{\mathbbm{1}}

\newcommand{\beq}{\begin{eqnarray*}}
\newcommand{\eeq}{\end{eqnarray*}}
\newcommand{\beqn}{\begin{eqnarray}}
\newcommand{\eeqn}{\end{eqnarray}}
\newcommand{\ben}{\begin{enumerate}}
\newcommand{\een}{\end{enumerate}}
\newcommand{\bit}{\begin{itemize}}
\newcommand{\eit}{\end{itemize}}
\newcommand{\hide}[1]{}

\newcommand{\eps}{\varepsilon}
\newcommand{\vertiii}[1]{{\left\vert\kern-0.25ex\left\vert\kern-0.25ex\left\vert #1 
    \right\vert\kern-0.25ex\right\vert\kern-0.25ex\right\vert}}

\renewcommand{\epsilon}{\eps}

\newcommand{\density}{\scal}

\newtheorem{definition}{Definition}
\newtheorem{lemma}{Lemma}
\newtheorem{remark}{Remark}
\newtheorem{proposition}[lemma]{Proposition}
\newtheorem{theorem}{Theorem}
\newtheorem{corollary}{Corollary}
\newtheorem{assumption}{Assumption}